\documentclass[final,twoside,11pt]{entics} 
\usepackage{thmtools}
\usepackage{thm-restate}
\usepackage{mathrsfs}
\usepackage{amsmath}
\usepackage{appendix}
\usepackage{adjustbox}
\usepackage{amssymb}
\usepackage{bbm}
\usepackage[utf8]{inputenc}
\usepackage{tikz}
\usepackage{tikz-cd}
\usepackage{caption}
\usepackage{subcaption}
\usepackage{algorithm}
\usepackage{algpseudocode}
\usepackage{stmaryrd}
\usepackage{ctable}
\usepackage{pgfplots}
\usetikzlibrary{automata,positioning,arrows}
\usepackage{enticsmacro}

\newcommand{\suff}{\textnormal{\texttt{suf}}}
\newcommand{\Exp}{\textnormal{Exp}}
\newcommand{\BA}{\textnormal{BA}}
\newcommand{\BExp}{\textnormal{BExp}}
\newcommand{\At}{\textnormal{At}}
\newcommand{\GS}{\textnormal{GS}}
\newcommand{\GSM}{\textnormal{GS}^{-}}
\newcommand{\false}{\texttt{false}}
\newcommand{\true}{\texttt{true}}
\newcommand{\assert}{\texttt{assert}}
\newcommand{\aand}{\texttt{and}}
\newcommand{\oor}{\texttt{or}}
\newcommand{\nnot}{\texttt{not}}
\newcommand{\ddo}{\texttt{do}}
\newcommand{\iif}{\texttt{if}}
\newcommand{\tthen}{\texttt{then}}
\newcommand{\eelse}{\texttt{else}}
\newcommand{\wwhile}{\texttt{while}}
\newcommand{\row}{row}
\newcommand{\GLStar}{\textnormal{\textsf{GL}}^{*}}
\newcommand{\LStar}{\textnormal{\textsf{L}}^{*}}
\newcommand{\Lsharp}{\textnormal{\textsf{L}}^{\sharp}}

\def\conf{MFPS 2022} 	
\volume{1}			
\def\volu{1}			
\def\papno{18}			
\def\mydoi{{\normalshape  \href{https://doi.org/10.46298/entics.10505}{10.46298/entics.10505}}}
\def\copynames{S. Zetzsche, A. Silva, M Sammartino} 
\def\runtitle{Guarded Kleene Algebra with Tests: Automata Learning}

\def\lastname{Zetzsche, Silva and Sammartino}
\begin{document}
\begin{frontmatter}
  \title{Guarded Kleene Algebra with Tests: Automata Learning} 
  \author{Stefan Zetzsche\thanksref{UCL}\thanksref{stefanfunding}\thanksref{myemail}}
  \author{Alexandra Silva\thanksref{Cornell}\thanksref{UCL}\thanksref{alexfunding}} 
  \author{Matteo Sammartino\thanksref{Royal}\thanksref{UCL}}
  \address[UCL]{University College London}  	
  \address[Cornell]{Cornell University}
  \address[Royal]{Royal Holloway, University of London}						
  \thanks[stefanfunding]{The author has been supported by GCHQ via the VeTSS grant “Automated black-box verification of networking systems” (4207703/RFA 15845) and by the ERC via the Consolidator Grant AutoProbe 101002697.} 
  \thanks[alexfunding]{The author has been supported by the ERC via the Consolidator Grant AutoProbe 101002697 and by a Royal Society Wolfson Fellowship.}   \thanks[myemail]{Email: \href{mailto:stefanzetzsche@gmail.com} {\texttt{\normalshape
        stefanzetzsche@gmail.com}}} 
\begin{abstract} 
  Guarded Kleene Algebra with Tests (GKAT) is the fragment of Kleene Algebra with Tests (KAT) that arises by replacing the union and iteration operations of KAT with predicate-guarded variants. GKAT is more efficiently decidable than KAT and expressive enough to model simple imperative programs, making it attractive for applications to e.g. network verification. In this paper, we further explore GKAT's automata theory, and present $\GLStar$, an algorithm for learning the GKAT automaton representation of a black-box, by observing its behaviour. A complexity analysis shows that it is more efficient to learn a representation of a GKAT program with $\GLStar$ than with Angluin's existing $\LStar$ algorithm. We implement $\GLStar$ and $\LStar$ in OCaml and compare their performances on example programs.
\end{abstract}
\end{frontmatter}

\section{Introduction}

As hardware and software systems continue to grow in size and complexity, practical and scalable methods for verification tasks become increasingly important. Classical model checking approaches to verification require the existence of a rich model of the system of interest, able to express all its relevant behaviour. In reality such a model however is rarely available, for instance, when the system comes in the form of a black-box with no access to the source code, or the system is simply too complex for manual processing. 

\emph{Automata learning}, or regular inference, aims to automatically infer an automata model by observing the behaviour of the system. The incremental approach has been successfully applied to a wide range of verification tasks from finding bugs in network protocols \cite{de2015protocol}, reverse engineering smartcard reader for internet banking \cite{chalupar2014automated}, and industrial applications \cite{hagerer2002model}. A comprehensive survey of the field can be found in \cite{vaandrager2017model}. The majority of modern learning algorithms is based on Angluin's $\LStar$ algorithm \cite{angluin1987learning}, which learns the unique minimal deterministic finite automaton (DFA) accepting a given regular language, or more generally, the unique minimal Moore automaton accepting a weighted language (\autoref{LStaralgorithm}). In many situations, however, targeting a DFA is not feasible, due to an explosion in the size of the state-space. Such cases instead require types of models specifically tailored for their domain-specific purposes.

For instance, modern networking systems can operate on very large data sets, making them very challenging to model. As a result, controlling, reasoning about, or extending networks can be surprisingly difficult. One approach to modernise the field that has recently gained popularity is \textit{Software Defined Networking} (SDN) \cite{feamster2014road}. Modern SDN programming languages, notably \emph{NetKAT} \cite{anderson2014netkat}, allow operators to model their network and dynamically fine tune forwarding behaviour in response to events such as traffic shifts.
Globally, NetKAT is based on \emph{Kleene Algebra} (KA) \cite{kozen1994completeness}, the sound and complete theory of regular expressions \cite{kleene1951representation}. Locally, it incorporates \emph{Boolean algebra}, the theory of predicates. Both logics have been unified in the well developed theory of \textit{Kleene Algebra with Tests} (KAT) \cite{kozen1997kleene}, which subsumes propositional Hoare logic and can be used to model standard imperative programming constructs. The automata theory for NetKAT has been introduced in \cite{foster2015coalgebraic}. 

Verifying properties about realistic networks reduces in NetKAT to deciding the behavioural equivalence of pairs of automata. Unfortunately, NetKAT's decision procedure is PSPACE-complete, mainly due its foundations in KAT.   
As a consequence, more efficiently decidable fragments of KAT have been considered. In \cite{smolka2019scalable} it was hinted that the \textit{guarded fragment} of KAT is notably more efficiently decidable than the full language, while still remaining sufficiently expressive for networking purposes. 
The idea has been taken further in \cite{smolka2019guarded}, which formally introduced \emph{Guarded Kleene Algebra with Tests} (GKAT), a variation on KAT that arises by replacing the union and iteration operations from KAT with guarded variants. In contrast to KAT, the equational theory of GKAT is decidable in (almost) linear time. These properties make GKAT a promising candidate for the foundations of a SDN programming language that is more efficiently decidable than NetKAT.

\begin{algorithm}[t]
	\begin{algorithmic}
		\State {$S,E \gets \lbrace \varepsilon \rbrace$}
		\Repeat 
				\While{$T = (S, E, \row: S \cup S \cdot A \rightarrow B^E)$ is not closed}	
					\State{find $t \in S \cdot A$ with $\row(t) \neq \row(s)$ for all $s \in S$}
					\State{$S \gets S \cup \lbrace t \rbrace $}
				\EndWhile
				\State {construct and submit $m(T)$ to the teacher}
				\If { the teacher replies \emph{no} with a counterexample $z \in A^*$ }
					\State { $E \gets E \cup \suff(z)$ }
				\EndIf
		\Until {the teacher replies \emph{yes}} \\
		\Return $m(T)$
	\end{algorithmic}
\caption{Angluin's $\LStar$ algorithm for Moore automata with input $A$ and output $B$}
\label{LStaralgorithm}
\end{algorithm}

In view of the potential applications of GKAT to the field of verification, this paper further investigates its automata theory. In detail, the paper makes the following contributions:
\begin{itemize}
	\item For any GKAT automaton, we define a second automaton, which we call its minimization (\autoref{minimdef}). We show that in the class of normal GKAT automata, the minimization of an automaton is the unique size-minimal normal automaton accepting the same language (\autoref{sizeminimal}). We show that the minimization of a normal automaton is isomorphic to the automaton that arises by identifying semantically equivalent pairs among reachable states (\autoref{minimalbisim}), and that the minimizations of two language equivalent normal automata are isomorphic (\autoref{minimalunique}). Finally, we show that minimizing a normal GKAT automaton preserves important invariants such as the nesting coequation (\autoref{minimizationcoequation}). 
	\item We present $\GLStar$, an active-learning algorithm (\autoref{GlStaralgorithm}) that incrementally infers a GKAT automaton from a black-box by querying an \emph{oracle} (\autoref{learningpart}). We show that if the oracle is instantiated with the language accepted by a finite normal GKAT automaton, then the algorithm terminates with its minimization in finite time (\autoref{correctnesstheorem}).
	\item We show that the semantics of GKAT automata \eqref{gkatautomatasemantics} can be reduced to the well-known semantics\footnote{In the language of Coalgebra, the semantics is given by the final coalgebra homomorphism for the functor defined by $FX = X^A \times B$, where $A = \At \cdot \Sigma = \lbrace \alpha \cdot p \mid \alpha \in \At,\ p \in \Sigma \rbrace$ and $B = 2^{\At}$, for finite sets $\Sigma$ and $\At$. The carrier of the final coalgebra for $F$ is $\mathcal{P}((\At \cdot \Sigma)^* \cdot \At)$, the set of \emph{guarded string languages}; the semantics of GKAT automata is given by the subclass of \emph{deterministic} guarded string languages.} of Moore automata. That is, there exists a language preserving embedding of GKAT automata into Moore automata (\autoref{embeddinglanguage}), which maps the minimization of a normal GKAT automaton to the language equivalent minimal Moore automaton (\autoref{minimalembeddingiso}). In consequence, GKAT programs could thus, in principle, be also represented by Moore automata, instead of GKAT automata.
	\item We present a complexity analysis which shows that for GKAT programs it is more efficient to learn a GKAT automaton representation with $\GLStar$ than a Moore automaton representation with Angluin's $\LStar$ algorithm (\autoref{complexity}). We implement $\GLStar$ and $\LStar$ in OCaml and compare their performances on example programs (\autoref{comparisongraph}).
	\end{itemize}

\section{Overview of the approach}

In this section, we give an overview of this paper through examples. We begin by presenting \autoref{LStaralgorithm}, a slight variation of Angluin's $\LStar$ algorithm for finite Moore automata. We exemplify the algorithm by executing it for the language semantics of a simple GKAT program. We then propose a new algorithm, which, instead of a Moore automaton, infers a GKAT automaton.

\subsection{$\LStar$ algorithm}

\begin{figure*}
\centering
\begin{subfigure}[b]{.08\textwidth}
	\centering
		\resizebox{0.9 \textwidth}{!}{
		\begin{tabular}{ c|c}
		 & $\varepsilon$ \\
		 \hline $\varepsilon$ & $0b + 0\overline{b}$   \\
		 \hline
		 \hline $b p$ & $0b + 0\overline{b}$  \\
		 \hline $b q$ & $0b + 0\overline{b}$  \\
		 \hline $\overline{b} p$ & $0b + 0\overline{b}$  \\
		 \hline $\overline{b} q$ & $1b + 1\overline{b}$ 
	\end{tabular}
	}
	\caption{}
		\label{LStarT1}
	\end{subfigure}
	\hfill
	\begin{subfigure}[b]{.1\textwidth}
	\centering
		\resizebox{0.8 \textwidth}{!}{
		\begin{tabular}{ c|c}
		 & $\varepsilon$  \\
		 \hline $\varepsilon$ & $0b + 0\overline{b}$   \\
		 \hline $\overline{b} q$ & $1b + 1\overline{b}$  \\
		 \hline
		 \hline $b p$ & $0b + 0\overline{b}$  \\
		 \hline $b q$ & $0b + 0\overline{b}$  \\
		 \hline $\overline{b} p$ & $0b + 0\overline{b}$  \\
		 \hline $\overline{b} q b p$ & $0b + 0\overline{b}$  \\
		 \hline $\overline{b} q b q$ & $0b + 0\overline{b}$  \\
		 \hline $\overline{b} q \overline{b} p$ & $0b + 0\overline{b}$  \\
		 \hline $\overline{b} q \overline{b} q$ & $0b + 0\overline{b}$  \\		 
	\end{tabular}
	}
	\caption{}
			\label{LStarT2}
	\end{subfigure}
	\hfill
	\begin{subfigure}[b]{.18\textwidth}
	\centering
			\resizebox{\textwidth}{!}{
				\begin{tikzpicture}[node distance=12em]
	\node[state, shape=circle, initial, initial text=, label=below:{$\Rightarrow 0b + 0\overline{b}$}] (x) {$\row(\varepsilon)$};
		\node[state, shape=circle, right of=x, label=above:{$\Rightarrow  1b + 1\overline{b}$}] (y) {$\row(\overline{b}q)$};
	    \path[->]
	(x) edge[above, bend left] node{$\overline{b}q$} (y)
	(x) edge[loop above] node{$bp, bq, \overline{b}p$} (x)
	(y) edge[below, bend left] node{$bp, bq, \overline{b}p, \overline{b}q$} (x)
	;
	\end{tikzpicture}
	}
	\caption{}
			\label{LStarmT2}
\end{subfigure}
	\hfill
	\begin{subfigure}[b]{.2\textwidth}
	\centering
	\resizebox{0.9 \textwidth}{!}{
		\begin{tabular}{ c|c|c|c }
		 & $\varepsilon$ & $\overline{b} q$ & $bq\overline{b}q$   \\
		 \hline $\varepsilon$ & $0b + 0\overline{b}$ & $1b + 1\overline{b}$ & $0b + 0\overline{b}$   \\
		 \hline $\overline{b} q$ & $1b + 1\overline{b}$ & $0b + 0\overline{b}$ & $0b + 0\overline{b}$ \\
		 \hline
		 \hline $b p$ & $0b + 0\overline{b}$ & $1b + 1\overline{b}$ & $0b + 0\overline{b}$  \\
		 \hline $b q$ & $0b + 0\overline{b}$ & $0b + 0\overline{b}$ & $0b + 0\overline{b}$ \\
		 \hline $\overline{b} p$ & $0b + 0\overline{b}$ & $0b + 0\overline{b}$ & $0b + 0\overline{b}$  \\
		 \hline $\overline{b} q b p$ & $0b + 0\overline{b}$ & $0b + 0\overline{b}$ & $0b + 0\overline{b}$ \\
		 \hline $\overline{b} q b q$ & $0b + 0\overline{b}$ & $0b + 0\overline{b}$ & $0b + 0\overline{b}$ \\
		 \hline $\overline{b} q \overline{b} p$ & $0b + 0\overline{b}$ & $0b + 0\overline{b}$ & $0b + 0\overline{b}$  \\
		 \hline $\overline{b} q \overline{b} q$ & $0b + 0\overline{b}$ & $0b + 0\overline{b}$ & $0b + 0\overline{b}$  \\		 
	\end{tabular}
	}
	\caption{}
				\label{LStarT3}
	\end{subfigure}
	\hfill
	\begin{subfigure}[b]{.2\textwidth}
	\centering
	\resizebox{0.9 \textwidth}{!}{
			\begin{tabular}{ c|c|c|c }
		 & $\varepsilon$ & $\overline{b} q$ & $bq\overline{b}q$   \\
		 \hline $\varepsilon$ & $0b + 0\overline{b}$ & $1b + 1\overline{b}$ & $0b + 0\overline{b}$   \\
		 \hline $\overline{b} q$ & $1b + 1\overline{b}$ & $0b + 0\overline{b}$ & $0b + 0\overline{b}$ \\
		 \hline $b q$ & $0b + 0\overline{b}$ & $0b + 0\overline{b}$ & $0b + 0\overline{b}$ \\
		 \hline
		 \hline $b p$ & $0b + 0\overline{b}$ & $1b + 1\overline{b}$ & $0b + 0\overline{b}$  \\
		 \hline $\overline{b} p$ & $0b + 0\overline{b}$ & $0b + 0\overline{b}$ & $0b + 0\overline{b}$  \\
		 \hline $\overline{b} q b p$ & $0b + 0\overline{b}$ & $0b + 0\overline{b}$ & $0b + 0\overline{b}$ \\
		 \hline $\overline{b} q b q$ & $0b + 0\overline{b}$ & $0b + 0\overline{b}$ & $0b + 0\overline{b}$ \\
		 \hline $\overline{b} q \overline{b} p$ & $0b + 0\overline{b}$ & $0b + 0\overline{b}$ & $0b + 0\overline{b}$  \\
		 \hline $\overline{b} q \overline{b} q$ & $0b + 0\overline{b}$ & $0b + 0\overline{b}$ & $0b + 0\overline{b}$  \\	
		 \hline $b q b p$ & $0b + 0\overline{b}$ & $0b + 0\overline{b}$ & $0b + 0\overline{b}$ \\
		 \hline $b q b q$ & $0b + 0\overline{b}$ & $0b + 0\overline{b}$ & $0b + 0\overline{b}$ \\
		 \hline $b q \overline{b} p$ & $0b + 0\overline{b}$ & $0b + 0\overline{b}$ & $0b + 0\overline{b}$  \\
		 \hline $b q \overline{b} q$ & $0b + 0\overline{b}$ & $0b + 0\overline{b}$ & $0b + 0\overline{b}$	 
	\end{tabular}	
	}
	\caption{}
					\label{LStarT4}
	\end{subfigure}
	\begin{subfigure}[b]{.2\textwidth}
	\centering
\label{minimalmoore}
			\resizebox{\textwidth}{!}{
\begin{tikzpicture}[node distance=6em]
	\node[state, shape=circle, initial, initial text=,  label=below:{$\Rightarrow 0b + 0\overline{b}$}] (x) {$\row(\varepsilon)$};
	\node[state, shape=circle, right of=x, above of=x, label=right:{$\Rightarrow  0b + 0\overline{b}$}] (y) {$\row(bq)$};
	\node[state, shape=circle, right of=y, below of=y, label=below:{$\Rightarrow  1b + 1\overline{b}$}] (z) {$\row(\overline{b}q)$};
	    \path[->]
	(x) edge[loop above] node{$b p$} (x)
	(x) edge[below] node{$\overline{b} q$} (z)
	(y) edge[loop above] node{$bp, bq, \overline{b}p, \overline{b}q$} (y)
	(z) edge[right] node{$bp, bq, \overline{b}p, \overline{b}q$}(y)
	(x) edge[right] node{$bq, \overline{b}p$} (y)
	;
\end{tikzpicture}
}
	\caption{}
				\label{LStarmT4}
\end{subfigure}
\caption{An example run of Angluin's $\LStar$ algorithm for the target language $\llbracket (\wwhile\ b\ \ddo\ p); q \rrbracket$.}
\label{lstarexamplerun}
\end{figure*}

Angluin's $\LStar$ algorithm learns the minimal DFA accepting a given regular language \cite{angluin1987learning}. The algorithm has since been modified and generalised for a broad class of transition systems. The variation we present here step-wise infers the minimal Moore automaton accepting a generalised language $L: A^* \rightarrow B$ for a finite input set $A$ and a finite output set $B$ \cite{moore1956gedanken}. The algorithm assumes the existence of a \emph{teacher} (or \emph{oracle}), which can respond to two types of queries:
\begin{itemize}
	\item \textbf{Membership queries}, consisting of a word $w \in A^*$, to which the teacher returns the output $L(w) \in B$; 
	\item \textbf{Equivalence queries}, consisting of a hypothesis Moore automaton $H$, to which the teacher responds \emph{yes}, if $H$ accepts $L$, and \emph{no} otherwise, providing a counterexample $z \in A^*$ in the symmetric difference of $L$ and the language accepted by $H$.
\end{itemize}
The algorithm incrementally builds an \emph{observation table}, which contains partial information about the language $L$ obtained by performing membership queries.
A table consists of two parts: a top part, with rows indexed by a finite set $S \subseteq A^*$; and a bottom-part, with rows ranging over $S \cdot A$. Columns are indexed by a finite set $E \subseteq A^*$. For any $t \in S \cup S \cdot A$ and $e \in E$, the entry at row $t$ and column $e$, denoted by $\row(t)(e)$, is given by the output $L(te) \in B$.
Note that the sets $S$ and $S \cdot A$ can intersect. In such a case, elements in the intersection are only shown in the top part. Formally, we refer to a table as a tuple $T = (S, E, \row)$, leaving the language $L$ implicit. 

	Given a table $T$, one can construct a Moore automaton $m(T) = (X, \delta, \varepsilon, x)$, where $X = \lbrace \row(s) \mid s \in S \rbrace$ is a finite set of states; the transition function $\delta: X \rightarrow X^A$ is given by $\delta(\row(s), a) = \row(sa)$; the output function $\varepsilon: X \rightarrow B$ satisfies $\varepsilon(\row(s)) = \row(s)(\varepsilon)$ (we abuse notation by writing $\varepsilon$ both for the empty string and for the output function); and
		$x = \row(\varepsilon)$ is the initial state.
	For $m(T)$ to be well-defined, the table $T$ has to satisfy $\varepsilon \in S$ and $\varepsilon \in E$, and two properties called closedness and consistency. An observation table is \emph{closed} if for all $t \in S \cdot A$ there exists an $s \in S$ such that $\row(t) = \row(s)$. An observation table is \emph{consistent}, if whenever $s, s' \in S$ satisfy $\row(s) = \row(s')$, then $\row(sa) = \row(s'a)$ for all $a \in A$. A table is consistent in particular if the function $\row$ is injective. 

The algorithm incrementally updates the table to satisfy those properties. If a well-defined hypothesis $m(T)$ can be constructed, the algorithm poses an equivalence query to the teacher, and either terminates, or refines the hypothesis with a counterexample $z \in A^*$. Since we respond to a negative equivalence query by adding the suffixes\footnote{The set $\suff(z)$ of suffixes for $z \in A^*$ is defined by $\suff(\varepsilon) = \lbrace \varepsilon \rbrace$ and $\suff(aw) = \lbrace aw \rbrace \cup \suff(w)$.}  of a counterexample to the set $E$ (opposed to adding the prefixes of a counterexample to the set $S$), rows will always be distinct, rendering consistency trivial\footnote{This variation of $\LStar$ has been introduced by Maler and Pnueli \cite{maler1995learnability}.}. At all times, the set $S$ is prefix-closed and the set $E$ is suffix-closed\footnote{A set $X \subseteq A^*$ is called \emph{suffix-closed}, if $\suff(z) \subseteq X$ for all $z \in X$.}.

\subsubsection{Example of execution}

We now execute Angluin's $\LStar$ (\autoref{LStaralgorithm}) for the target language
\begin{equation}
\label{acceptedlanguageexample}
	L = \llbracket (\wwhile\ b\ \ddo\ p); q \rrbracket =   \lbrace \overline{b}qb, \overline{b}q \overline{b}, b p \overline{b} q b, b p \overline{b} q \overline{b}, ... \rbrace  \subseteq (\At \cdot \Sigma)^* \cdot \At,
	\end{equation}
	where $\At = \lbrace b, \overline{b} \rbrace$ is a finite set of \emph{atoms} and $\Sigma = \lbrace p, q \rbrace$ is a finite set of \emph{actions}. The language $L$ represents the semantics of a program that performs the action $p$  while $b$ is true, and otherwise continues with $q$. It can be viewed as a generalised language $\widehat{L}$ with input $A = (\At \cdot \Sigma)$ and output $B = 2^{\At}$ via currying. We denote functions $f \in B$ as formal sums $\sum_{\alpha \in \At} f(\alpha)\alpha$. Each query to $\widehat{L}$ requires $\vert \At \vert$ many queries to $L$.
	
	Initially, the sets $S$ and $E$ are set to the singleton $\lbrace \varepsilon \rbrace$. We build the observation table in \autoref{LStarT1}. Since the row indexed by $\overline{b}q$ does not appear in the upper part, i.e. differs from the row indexed by $\varepsilon$, the table is not closed.
		To resolve the closedness defect we add $\overline{b}q$ to $S$. The observation table (\autoref{LStarT2}) is now closed. We derive from it the hypothesis depicted in \autoref{LStarmT2}. 
		Next, we pose an equivalence query, to which the oracle replies \emph{no} and informs us that the word $z = bq \overline{b}q$ has been falsely classified. Indeed, given $z$, the language accepted by the hypothesis outputs $1b + 1 \overline{b}$, whereas \eqref{acceptedlanguageexample} produces  $0b + 0\overline{b}$.	
To respond to the counterexample $z$, we add its suffixes to $E$. In this case, there are only the two suffixes $\overline{b} q$ and $bq\overline{b}q$. The next observation table (\autoref{LStarT3}) again is not closed: the row indexed by e.g. $bq$ does not equal any of the two upper rows indexed by $\varepsilon$ and $\overline{b}q$.
	To resolve the closedness defect we add $bq$ to $S$, and obtain the table in \autoref{LStarT4}.
The observation table is now closed. We derive from it the automaton in \autoref{LStarmT4}. Next, we pose an equivalence query, to which the oracle replies \emph{yes}.

\subsection{$\GLStar$ algorithm}

	\begin{algorithm}[t]
	\begin{algorithmic}
		\State {$S \gets \lbrace \varepsilon \rbrace, E \gets \At $}
		\Repeat 
				\While{$T = (S, E, \row: S \cup S \cdot (\At \cdot \Sigma)\rightarrow 2^E)$ is not closed}	
					\State{find $t \in S \cdot (\At \cdot \Sigma)$ with $\row(t)(e) = 1$ for some $e \in E$, but $\row(t) \neq \row(s)$ for all $s \in S$}
					\State{$S \gets S \cup \lbrace t \rbrace $}
				\EndWhile
				\State {construct and submit $m(T)$ to the teacher}
				\If { the teacher replies \emph{no} with a counterexample $z \in (\At \cdot \Sigma)^* \cdot \At$ }
					\State { $E \gets E \cup \suff(z)$ }
				\EndIf
		\Until {the teacher replies \emph{yes}} \\
		\Return $m(T)$
	\end{algorithmic}
\caption{The $\GLStar$ algorithm for GKAT automata}
\label{GlStaralgorithm}
\end{algorithm}

In this section, we propose a new algorithm (\autoref{GlStaralgorithm}) for learning GKAT program representations, which we call $\GLStar$. The new algorithm modifies \autoref{LStaralgorithm} by addressing a number of observations. 

First, we note that the Moore automaton in \autoref{LStarmT4} admits multiple transitions to $\row(bq)$, a \emph{sink-state}, which does not accept any words.
	Second, we observe that languages induced by GKAT programs are \emph{deterministic}\footnote{Deterministic in the sense that, whenever two strings agree on the first $n$ atoms, then they agree on their first $n$ actions (or lack thereof).}. Such languages are naturally represented by GKAT automata, which keep some transitions implicit.
	Third, in some cases\footnote{For instance, the entries of the row indexed by $bq$ in \autoref{LStarT3} must all be zero, since the row indexed by $bp$ admits a non-zero entry.} the deterministic nature of the target language allows us to fill-in parts of the observation table without performing any membership queries.
	Fourth, the cells of the observation table are labelled by functions, each of which requires two membership queries to \eqref{acceptedlanguageexample}; as a consequence, table extensions require an unfeasible amount of queries.

As before, we assume two finite sets, $\At$ and $\Sigma$, and a deterministic language $L \subseteq (\At \cdot \Sigma)^* \cdot \At$. The oracle of $\GLStar$ can answer two types of queries: membership queries consist of a word $w \in (\At \cdot \Sigma)^* \cdot \At$, to which the oracle returns the output $L(w) \in 2$; equivalence queries consist of a hypothesis GKAT automaton $H$, to which the oracle responds \emph{yes}, if $H$ accepts $L$, and \emph{no} otherwise, providing a counterexample $z \in (\At \cdot \Sigma)^* \cdot \At$ in the symmetric difference of $L$ and the language accepted by $H$.

An observation table in $\GLStar$ consists of two parts: a top part, with rows indexed by a finite set $S \subseteq (\At \cdot \Sigma)^*$; and a bottom-part, with rows ranging over $S \cdot \At \cdot \Sigma$. Columns range over a finite set $E \subseteq (\At \cdot \Sigma)^* \cdot \At$. The entry of the observation table at row $t$ and column $e$, denoted by $\row(t)(e)$, is given by $L(te) \in 2$. We refer to a table by $T= (S, E, \row)$ and leave the deterministic language $L$ implicit.

Given an observation table $T$, we construct a GKAT automaton $m(T) = (X, \delta, x)$, where $X = \lbrace \row(s) \mid s \in S \rbrace$ is a finite set of states; $x = \row(\varepsilon)$ is the initial state; and $\delta: X \rightarrow (2 + \Sigma \times X)^{\At}$ is the transition function which evaluates $\delta(\row(s))(\alpha)$ to $(p, \row(s \alpha p))$, if there exists an $e \in E$ with $\row(s \alpha p)(e) = 1$; to $1$, if $ \row(s)(\alpha) = 1$; and to $0$, otherwise.

Most of the properties a table needs to satisfy such that the hypothesis $m(T)$ is well-defined are guaranteed by the construction of \autoref{GlStaralgorithm}, since $L$ is deterministic. We only have to verify that the table is \emph{closed}, that is, for all $t \in S \cdot \At \cdot \Sigma$ with $\row(t)(e) = 1$ for some $e \in E$, there exists some $s \in S$ such that $\row(t) = \row(s)$. As in the case of $\LStar$, the algorithm incrementally updates the table until closedness is guaranteed. It then constructs a well-defined hypothesis, and poses an equivalence query to the teacher. If the oracle replies \emph{yes}, the algorithm terminates, and if the response is \emph{no}, it adds the suffixes\footnote{The set $\suff(z)$ of suffixes for $z \in A^* \cdot B$ is defined by $\suff(w b) = \lbrace vb \mid v \in \suff(w) \rbrace$.}  of a counterexample $z \in (\At \cdot \Sigma)^* \cdot \At$ to $E$.

The differences between $\GLStar$ and $\LStar$ (instantiated for $A = \At \cdot \Sigma$ and $B = 2^{\At}$) are essentially a consequence of currying. In the former case, the set $E$ contains elements of type $(\At \cdot \Sigma)^* \cdot \At$, and the table is filled with booleans in $2$; in the latter case, the set $E$ contains elements of type $(\At \cdot \Sigma)^*$, and the table is filled with functions $\At \rightarrow 2$. This, however, does not mean that $\GLStar$ is merely a shift in perspective: its new types induce independent definitions, and termination needs to be established with novel correctness proofs (\autoref{learningpart}). A thorough comparison with $\LStar$ is given in \autoref{comparisonmoore}.
   \begin{figure*}
\centering
\begin{subfigure}[b]{.13\textwidth}
	\centering
		\resizebox{0.6 \textwidth}{!}{
		\begin{tabular}{ c|c|c}
		 & $b$ & $\overline{b}$ \\
		 \hline $\varepsilon$ & 0 & 0  \\
		 \hline
		 \hline $b p$ & 0 & 0 \\
		 \hline $b q$ & 0 & 0 \\
		 \hline $\overline{b} p$ & 0 & 0 \\
		 \hline $\overline{b} q$ & 1 & 1
	\end{tabular}
	}
	\caption{}
	\label{GLStarT1}
	\end{subfigure}
	\hfill
	\begin{subfigure}[b]{.15\textwidth}
	\centering
		\resizebox{0.6 \textwidth}{!}{
		\begin{tabular}{ c|c|c}
		 & $b$ & $\overline{b}$ \\
		 \hline $\varepsilon$ & 0 & 0  \\
		 \hline $\overline{b} q$ & 1 & 1 \\
		 \hline
		 \hline $b p$ & 0 & 0 \\
		 \hline $b q$ & 0 & 0 \\
		 \hline $\overline{b} p$ & 0 & 0 \\
		 \hline $\overline{b} q b p$ & 0 & 0 \\
		 \hline $\overline{b} q b q$ & 0 & 0 \\
		 \hline $\overline{b} q \overline{b} p$ & 0 & 0 \\
		 \hline $\overline{b} q \overline{b} q$ & 0 & 0 \\		 
	\end{tabular}
	}
	\caption{}
		\label{GLStarT2}
	\end{subfigure}
	\hfill
		\begin{subfigure}[b]{.22\textwidth}
	\centering
	\resizebox{0.8 \textwidth}{!}{
	\begin{tikzpicture}[node distance=7em]
	\node[state, shape=circle, initial, initial text=, label=above:{$\Rightarrow b \mid 0$}] (x) {$\row(\varepsilon)$};
		\node[state, shape=circle, right of=x, label=above:{$\Rightarrow  b, \overline{b} \mid 1$}] (y) {$\row(\overline{b}q)$};
	    \path[->]
	(x) edge[above] node{$\overline{b} \mid q$} (y)
	;
	\end{tikzpicture}
	}
	\caption{}
		\label{GLStarmT2}
	\end{subfigure}
	\hfill
	\begin{subfigure}[b]{.26\textwidth}
	\centering
	\resizebox{0.6 \textwidth}{!}{
		\begin{tabular}{ c|c|c|c|c }
		 & $b$ & $\overline{b}$ & $bp\overline{b}qb$ & $\overline{b}qb$  \\
		 \hline $\varepsilon$ & 0 & 0 & 1 & 1  \\
		 \hline $\overline{b} q$ & 1 & 1 & 0 & 0\\
		 \hline
		 \hline $b p$ & 0 & 0 & 1 & 1 \\
		 \hline $b q$ & 0 & 0 & 0 & 0\\
		 \hline $\overline{b} p$ & 0 & 0 & 0 & 0 \\
		 \hline $\overline{b} q b p$ & 0 & 0 & 0 & 0 \\
		 \hline $\overline{b} q b q$ & 0 & 0 & 0 & 0\\
		 \hline $\overline{b} q \overline{b} p$ & 0 & 0 & 0 & 0 \\
		 \hline $\overline{b} q \overline{b} q$ & 0 & 0 & 0 & 0 \\		 
	\end{tabular}
	}
	\caption{}
		\label{GLStarT3}
	\end{subfigure}
	\hfill
	\begin{subfigure}[b]{.21\textwidth}
	\centering
	\resizebox{0.8 \textwidth}{!}{
	\begin{tikzpicture}[node distance=7em]
	\node[state, shape=circle, initial, initial text=] (x) {$\row(\varepsilon)$};
		\node[state, shape=circle, right of=x, label=above:{$\Rightarrow  b, \overline{b} \mid 1$}] (y) {$\row(\overline{b}q)$};
	    \path[->]
	(x) edge[loop above] node{$b \mid p$} (x)
	(x) edge[above] node{$\overline{b} \mid q$} (y)
	;
	\end{tikzpicture}
	}
	\caption{}
		\label{GLStarmT3}
	\end{subfigure}
\caption{An example run of $\GLStar$ for the target language $\llbracket (\wwhile\ b\ \ddo\ p); q \rrbracket$.}
\label{glstarexamplerun}
\end{figure*}

\subsubsection{Example of execution}

\label{glstarexamplerunsection}

We now execute \autoref{GlStaralgorithm} for the target language \eqref{acceptedlanguageexample}.
Initially, $S = \lbrace \varepsilon \rbrace$ and $E = \At$.	
		We build the observation table in \autoref{GLStarT1}.  Since the bottom row indexed by $\overline{b}q$ contains a non-zero entry and differs from all upper rows (in this case, only the row indexed by $\varepsilon$), the table is not closed. We resolve the closedness defect by adding  $\overline{b}q$ to $S$. The observation table (\autoref{GLStarT2}) is now closed. Note that the row indexed by $\overline{b}q$ indicates that the words $\overline{b}qb$ and $\overline{b}q\overline{b}$ are accepted. Since we know the target language is deterministic, the last four rows of the table can be filled with zeroes, without performing any membership queries. From \autoref{GLStarT2} we derive the hypothesis depicted in \autoref{GLStarmT2}. Next, we pose an
equivalence query, to which the oracle replies $\emph{no}$ and provides us with the counterexample $z = bp\overline{b}qb$, which is in the language \eqref{acceptedlanguageexample}, but not accepted by the hypothesis.	
	We respond to the counterexample by adding its suffixes $bp\overline{b}qb$, $\overline{b}qb$ and $b$ to $E$. The resulting observation table is depicted in \autoref{GLStarT3}.
	The table is closed, since the only non-zero bottom row is the one indexed by $bp$, which coincides with the upper row indexed by $\varepsilon$. Since the row indexed by $bp$ has a non-zero entry, the row indexed by $bq$ can automatically be filled with zeroes.
	We derive from \autoref{GLStarT3} the automaton in \autoref{GLStarmT3}. Finally, we pose an equivalence query, to which the oracle replies \emph{yes}.	

\section{Preliminaries}
This section introduces the syntax and semantics of GKAT, an abstract imperative programming language with uninterpreted actions. For most parts, we follow the relevant bits of the original presentation in \cite{smolka2019guarded}.

\subsection{Syntax}

The syntax of GKAT is inductively built from disjoint non-empty sets of \emph{primitive tests}, $T$, and \emph{actions}, $\Sigma$. In a first step, one generates from $T$ a set of Boolean expressions, $\BExp$. In a second step, the set is extended with $\Sigma$, to the full set of GKAT expressions, $\Exp$:
\begin{align*}
	b,c,d \in \BExp &::= 0 \mid 1 \mid t \in T \mid b \cdot c \mid b + c \mid \overline{b} \\
e, f, g \in \Exp &::= p \in \Sigma \mid b \in \BExp \mid e \cdot f \mid e +_b f \mid e^{(b)}
\end{align*}
By a slight abuse of notation, we will sometimes write $e f$ for $e \cdot f$ and keep parenthesis implicit, e.g. $bc + d$ should be read as $(b \cdot c) + d$.

It is natural to view GKAT expressions as uninterpreted imperative programs. Under this view, one makes the identifications depicted in \autoref{gkatexpressionsasprograms}.

\begin{figure*}
	\begin{gather*}
	0 \equiv \false \qquad 1 \equiv \true \qquad t \equiv t \qquad 
	b \cdot c \equiv b\ \aand\ c \qquad b + c \equiv b\ \oor\ c  \qquad \overline{b} \equiv \nnot\ b \\
	p \equiv \ddo\ p \qquad b \equiv \assert\ b
	 \qquad
	 e \cdot f \equiv e; f \qquad e^{(b)} \equiv \wwhile\ b\ \ddo\ e \qquad e +_b f \equiv \iif\ b\ \tthen\ e\ \eelse\ f
\end{gather*}
\caption{Identifying GKAT expressions with imperative programs.}
\label{gkatexpressionsasprograms}
\end{figure*}

Readers familiar with KAT will notice that the grammar for GKAT is similar to the one of KAT. It differs in that GKAT replaces KAT's union $(+)$ with the guarded union $(+_b)$, and KAT's iteration $(e^*)$ with the guarded iteration $(e^{(b)})$. GKAT's expressions can be encoded within KAT's grammar via the standard embedding that maps a conditional $e +_b f$ to $be + \overline{b}f$, and a while-loop $e^{(b)}$ to $(be)^*\overline{b}$.

\subsection{Semantics: Language Model}

In this section, we introduce the language semantics of GKAT, which assigns to a program the traces it could produce once executed. 
Intuitively, an execution trace is a string of the shape $\alpha_0 p_1 \alpha_1 ... p_n \alpha_n$. It can be thought of as a sequence of states $\alpha_i$ a system is in at point $i$ in time, beginning with $\alpha_0$ and ending in $\alpha_n$, intertwined with actions $p_i$ that transition from the state $\alpha_{i-1}$ to the state $\alpha_{i}$.

More formally, let $\equiv_{\BA}$ denote the equivalence relation between Boolean expressions induced by the Boolean algebra axioms. The quotient $\BExp/_{\equiv_{\BA}}$, that is, the free Boolean algebra on generators $T$, admits a natural preorder $\leq$ defined by $b \leq c \Leftrightarrow b + c \equiv_{\BA} c$. The minimal nonzero elements with respect to this order are called \emph{atoms}, the set of which is denoted by $\At$. If $T = \lbrace t_1, ..., t_n \rbrace$ is finite, an atom $\alpha \in \At$ is of the form $\alpha = c_1 \cdot ... \cdot c_n$ with $c_i \in \lbrace t_i, \overline{t_i} \rbrace$.

A \emph{guarded string} is an element of the set $\GS := \At \cdot (\Sigma \cdot \At)^*$, or equivalently, $(\At \cdot \Sigma)^* \cdot \At$. The set of guarded strings without terminating atom is $\GSM := (\At \cdot \Sigma)^*$. 

A guarded string language $L \subseteq \GS$ is \emph{deterministic} \cite[Def. 2.2]{smolka2019guarded}, if, whenever $\alpha_1 p_1 ... \alpha_{n-1} p_{n-1}\alpha_n v \in L$ and $\alpha_1 q_1 ... \alpha_{n-1}q_{n-1}\alpha_n w \in L$, then $p_i = q_i$ for all $1 \leq i \leq n-1$, and either $v = w = \varepsilon$, or $v = p_n v'$ and $w = q_n w'$ with $p_n = q_n$. The set of deterministic guarded string languages is denoted by $\mathscr{L}$.

Guarded strings can be partially composed via the \emph{fusion product} defined by $v \alpha \diamond \beta w := v \alpha w$, if $\alpha = \beta$, and undefined otherwise. The partial product lifts to a total function on guarded languages by $L \diamond K := \lbrace v \diamond w \mid v \in L, w \in K \rbrace$.
The $n$-th power of a guarded language is inductively defined by $L^0 := \At$ and $L^{n+1} := L^n \diamond L$.
For $B \subseteq \At$ and $\overline{B} := \At \setminus B$, the guarded sum and the guarded iteration of languages are given by
\[
L +_B K := (B \diamond L) \cup (\overline{B} \diamond K) \qquad L^{(B)} := \cup_{n \geq 0} (B \diamond L)^n  \diamond \overline{B}.
\]
The \emph{language model} of GKAT is given by the semantic function $\llbracket - \rrbracket : \Exp \rightarrow \mathscr{P}(\GS)$, which is inductively defined as follows:  
\begin{gather*}
	\llbracket p \rrbracket := \lbrace \alpha p \beta \mid \alpha, \beta \in \At \rbrace \qquad \llbracket b \rrbracket := \lbrace \alpha \in \At \mid \alpha \leq b \rbrace  \\
	\llbracket e \cdot f \rrbracket := \llbracket e \rrbracket \diamond \llbracket f \rrbracket \qquad \llbracket e +_b f \rrbracket := \llbracket e \rrbracket +_{\llbracket b \rrbracket} \llbracket f \rrbracket \qquad
	\llbracket e^{(b)} \rrbracket := \llbracket e \rrbracket^{(\llbracket b \rrbracket)}. 
\end{gather*}
Equivalently, the language semantics of GKAT can be constructed by post-composing the embedding of GKAT expressions into KAT expressions with the semantics of KAT.

The guarded string language $\llbracket e \rrbracket$ accepted by a GKAT program $e$ is deterministic.

\begin{example}
\label{examplewhileloop}
Let the sets of primitive tests and actions be defined by $T := \lbrace b \rbrace$ and $\Sigma := \lbrace p , q\rbrace$, respectively. Then there exist only two atoms, $\At = \lbrace b, \overline{b} \rbrace$. The language model assigns to the program
$ p^{(b)}q \equiv (\wwhile\ b\ \ddo\ p); q$ the guarded deterministic language \eqref{acceptedlanguageexample}.
\end{example}

\subsection{Semantics: Automata Model}

\begin{figure}
\centering
\adjustbox{scale=0.8}{
		\begin{tikzpicture}[node distance=4em]
	\node[state, shape=circle, initial, initial text=] (x) {$x$};
		\node[state, shape=circle, left of=x, below of=x] (y) {$y$};
		\node[state, shape=circle, right of=x, below of=x, label=right:{$\Rightarrow  b, \overline{b} \mid 1$}] (z) {$z$};
	    \path[->]
	 (x) edge[left] node{$b \mid p$} (y)
	 (x) edge[right] node{$\overline{b} \mid q$} (z)
	(y) edge[loop left] node{$b \mid p$} (y)
	(y) edge[above] node{$\overline{b} \mid q$} (z)
	;
	\end{tikzpicture}
	}
	\caption{The Thompson-automaton $\mathscr{X}_{p^{(b)}q}$ for $T = \lbrace b \rbrace$ and $\Sigma = \lbrace p , q\rbrace$.}
			\label{thompsonautomaton}
\end{figure}

In this section, we introduce the automata model of GKAT, the central subject of this paper. As before, we assume two finite sets of tests $T$ and actions $\Sigma$, the former of which induces a finite set of atoms, $\At$.  

Let $G$ be the functor on the category of sets which is defined on objects by $GX = (2 + \Sigma \times X)^{\At}$, where $2 = \lbrace 0, 1 \rbrace$ is the two-element set, and on morphisms in the usual way. 
A \emph{$G$-coalgebra} consists of a pair $\mathscr{X} = (X, \delta)$, where $X$ is a set called \emph{state-space} and $\delta: X \rightarrow GX$ is a function called \emph{transition map}. A $G$-coalgebra \emph{homomorphism} $f: (X, \delta^X) \rightarrow (Y, \delta^Y)$ is a function $f: X \rightarrow Y$  that commutes with the transition maps, $\delta^{Y} \circ f = Gf \circ \delta^X$. A \emph{$G$-automaton} is a $G$-coalgebra $\mathscr{X}$ with a designated initial state $x \in X$. A homomorphism $f: (\mathscr{X}, x) \rightarrow (\mathscr{Y}, y)$ between $G$-automata is a homomorphism between the underlying $G$-coalgebras that maps initial state to initial state, $f(x) = y$. 

For each state $x \in X$, given an input $\alpha \in \At$, a $G$-coalgebra either i) halts and accepts, that is, satisfies $\delta(x)(\alpha) = 1$; ii) halts and rejects, that is, satisfies $\delta(x)(\alpha) = 0$; or iii) produces an output $p$ and moves to a new state $y$, that is, satisfies $\delta(x)(\alpha) = (p, y)$. Intuitively, for each state $x \in X$, a guarded string $\alpha_0p_1\alpha_1...p_n \alpha_n$ is accepted, if the $G$-coalgebra in state $x$ produces the output $p_1...p_n$, halts and accepts. Formally, one defines a function $\llbracket - \rrbracket: X \rightarrow \mathscr{P}(\GS)$ as follows:
		\begin{align}
		\label{gkatautomatasemantics}
		\begin{split}
			\alpha \in \llbracket x \rrbracket :\Leftrightarrow \delta(x)(\alpha) = 1; \qquad
			\alpha p w \in \llbracket x \rrbracket :\Leftrightarrow \exists y \in X: \delta(x)(\alpha) = (p, y) \textnormal{ and } w \in  \llbracket y \rrbracket.
		\end{split}	
		\end{align}
		
								A $G$-coalgebra is \emph{observable}, if the function $\llbracket - \rrbracket$ is injective.
								
	A guarded string $w \in \GS$ is \emph{accepted} by a state $x \in X$, if $w \in \llbracket x \rrbracket$. The language accepted by a $G$-automaton, $\llbracket \mathscr{X} \rrbracket$, is the language accepted by its initial state. 
	Every language accepted by a $G$-automaton satisfies the determinacy property \cite[Thm. 5.8]{smolka2019guarded}.
	 Conversely, one can equip the set of deterministic languages with a $G$-coalgebra structure $(\mathscr{L}, \delta^{\mathscr{L}})$ defined by
\begin{align*}
			\delta^{\mathscr{L}}(L)(\alpha) = \begin{cases}
				(p, (\alpha p)^{-1}L) & \textnormal{ if } (\alpha p)^{-1}L \not = \emptyset \\
				1 & \textnormal{ if } \alpha \in L \\
				0 & \textnormal{ otherwise}
			\end{cases},
		\end{align*}
		where $(\alpha p)^{-1}L = \lbrace w \in \GS \mid \alpha p w \in L \rbrace$.
	Since $\llbracket L \rrbracket = L$ for any $L \in \mathscr{L}$ \cite[Thm. 5.8]{smolka2019guarded}, every deterministic language can thus be recognized by a $G$-automaton with possibly infinitely many states. 
	
	A $G$-coalgebra $(X, \delta)$ is \emph{normal}, if it only transitions to \emph{live} states, that is, $\delta(x)(\alpha) = (p, y)$ implies $\llbracket y \rrbracket \not = \emptyset$, for all $x, y \in X$. For any $G$-automaton $\mathscr{X}$ one can construct a language equivalent normal $G$-automaton $\widehat{\mathscr{X}}$ \cite[Lem. 5.6]{smolka2019guarded}. If $\mathscr{X}$ is normal, the function $\llbracket - \rrbracket: X \rightarrow \mathscr{P}(\GS)$ is the unique coalgebra homomorphism $\llbracket - \rrbracket: (X, \delta) \rightarrow (\mathscr{L}, \delta^{\mathscr{L}})$ \cite[Thm. 5.8]{smolka2019guarded}. 
	
	Two states $x,y \in X$ of a normal coalgebra accept the same language, $\llbracket x \rrbracket = \llbracket y \rrbracket$, if and only if they are \emph{bisimilar}, $x \simeq y$, that is, there exists a binary relation $R \subseteq X \times X$, such that, if $x R y$, then it holds:
	\begin{itemize}
		\item if $\delta(x)(\alpha) \in 2$, then $\delta(y)(\alpha) = \delta(x)(\alpha)$; and
		\item if $\delta(x)(\alpha) = (p, x')$, then $\delta(y)(\alpha) = (p, y')$ and $x' R y'$ for some $y' \in X$.
	\end{itemize}
	Bisimilarity is a symmetric relation and can be extended to two coalgebras by constructing a coalgebra that has the disjoint union of their state-spaces as state-space.

Using a construction that is reminiscent of Thompson's construction for regular expressions \cite{thompson1968programming}, it is possible to efficiently interpret a GKAT expression $e$ as an automaton $\mathscr{X}_e$ that accepts the same language \cite{smolka2019guarded}. Alternatively, one can mirror \cite{smolka2019guarded} Kozen's syntactic form of Brzozowski's derivatives for KAT \cite{kozen2017coalgebraic}. 
\begin{example}
\label{thompsonexample}
	The Thompson-automaton assigned to the expression $p^{(b)}q \equiv (\wwhile\ b\ \ddo\ p); q$ is depicted in \autoref{thompsonautomaton}. It is normal, but not observable, since the states $x$ and $y$ are bisimilar, $x \simeq y$, thus accept the same language, $\llbracket x \rrbracket = \llbracket y \rrbracket$. Moreover, it is language equivalent to the expression by which it is generated, that is, it  satisfies $\llbracket \mathscr{X}_{p^{(b)}q} \rrbracket = \llbracket p^{(b)}q \rrbracket$.
\end{example}

\section{The minimal representation $m(\mathscr{X})$}

The automaton $\mathscr{X}_e$ assigned to an expression $e$ by the Thompson construction is not always the most efficient representation of the language $\llbracket e \rrbracket$. For instance, as seen in \autoref{thompsonexample}, the Thompson-automaton $\mathscr{X}_{p^{(b)}q}$ in \autoref{thompsonautomaton} contains redundant structure, since its states $x$ and $y$ exhibit the same behaviour. In this section, we show that any $G$-automaton $\mathscr{X}$ admits an equivalent \emph{minimal} representation, $m(\mathscr{X})$.
\subsection{Reachability}

We begin by formally defining what it means for a state of a $G$-automaton to be reachable, and show that restricting an automaton to its reachable states preserves important invariants.
 
\begin{definition}
	Let $(X, \delta)$ be a $G$-coalgebra. We write $\rightarrow\ \subseteq X \times \GSM \times X$ for the smallest relation satisfying:
\begin{gather}
\label{transitiondef}
	\frac{}{x \xrightarrow{\varepsilon} x} \quad
	\frac{\delta(x)(\alpha) = (p, y)}{x \xrightarrow{\alpha p} y} \quad
	\frac{x \xrightarrow{\alpha_1 p_1 ... \alpha_{n-1} p_{n-1}} y \quad , \quad y \xrightarrow{\alpha_n p_n} z}{x \xrightarrow{\alpha_1 p_1 ... \alpha_n p_n}z}.
\end{gather}
The states \emph{reachable} from $x \in X$ are 
$r(x) := \lbrace y \in X \mid \exists w \in \GSM: x \xrightarrow{w} y \rbrace$, and their \emph{witnesses} are $R(x) := \lbrace w \in \GSM \mid \exists x_w \in X : x \xrightarrow{w} x_w \rbrace$.
\end{definition}

The following result shows that a state reached by a word is uniquely defined.

\begin{restatable}{lemma}{uniquenessstates}
\label{uniquenessstates}
If $x \xrightarrow{w} x_w^1$ and $x \xrightarrow{w} x_w^2$, then $x_w^1 = x_w^2$.
\end{restatable}

It is not hard to see that the subset $r(x)$ of reachable states is $\delta$-invariant, i.e. if $y \in r(x)$ and $\delta(y)(\alpha) = (p, z)$, then $z \in r(x)$. We denote the well-defined sub-automaton one obtains by restricting to the states reachable from an initial state as $r(\mathscr{X})$, and call an automaton \emph{reachable}, if $\mathscr{X} = r(\mathscr{X})$. Following \cite[Def. 15]{van2017calf}, we call a normal, reachable, and observable automaton $\emph{minimal}$.

The set $R(x)$ of words witnessing the reachability of states in $\mathscr{X} = (X, \delta, x)$ can be equipped with a $G$-automaton structure $R(\mathscr{X}) := (R(x), \partial, \varepsilon)$, where $\partial(w)(\alpha) = (p, w\alpha p)$, if $\delta(x_w)(\alpha) = (p, x_{w \alpha p})$ for some $x_{w \alpha p} \in X$, and $\partial(w)(\alpha) = \delta(x_w)(\alpha)$ otherwise. The automaton $r(\mathscr{X})$ can then be recovered as the image of the automata homomorphism $f: R(\mathscr{X}) \rightarrow \mathscr{X}$ defined by $f(w) = x_w$. In other words, there exists an epi-mono factorization $R(\mathscr{X}) \twoheadrightarrow r(\mathscr{X}) \hookrightarrow \mathscr{X}$.

We conclude with a list of important properties preserved by restricting an automaton to its reachable states. \emph{Well-nestedness} and \emph{coequations}, in particular, the \emph{nesting coequation}, have been introduced in \cite{smolka2019guarded} and \cite{schmid2021guarded}, respectively. We refer the reader to the original papers for formal definitions, and to \autoref{relatedwork} for a high-level comparison.

\begin{restatable}{proposition}{reachablesubcoalgebra}
\label{reachablesubcoalgebra}
Let $\mathscr{X}$ be a $G$-automaton, then $r(\mathscr{X})$ is well-nested, normal, or satisfies the nesting coequation, whenever $\mathscr{X}$ does. Moreover, $r(\mathscr{X})$ accepts the same language as $\mathscr{X}$.
\end{restatable}

\subsection{Minimality}

Recall that the state-space of the minimal DFA for a regular language consists of the equivalence classes of the so-called Myhill-Nerode equivalence relation \cite{nerode1958linear}. 

Similarly, we define the state-space of the minimization of a GKAT automaton $\mathscr{X}$ as the equivalence classes of the equivalence relation $\equiv_{\llbracket \mathscr{X} \rrbracket}$ on $\GSM$ defined for any guarded string language $L \subseteq \GS$ by: 
	\begin{equation}
	\label{equivdef}
		v \equiv_L w :\Leftrightarrow \forall u \in \GS: vu \in L \textnormal{ if(f) } wu \in L.
	\end{equation}
	Let $v^{-1}L = \lbrace u \in \GS \mid v u \in L \rbrace$, then two words $v, w$ are equivalent with respect to $\equiv_L$ if(f) their derivatives $v^{-1}L$ and $w^{-1}L$ coincide.

\begin{definition}
\label{minimdef}
	The \emph{minimization} of a $G$-automaton $\mathscr{X} = (X, \delta, x)$ is $m(\mathscr{X}) := (\lbrace w^{-1}\llbracket \mathscr{X} \rrbracket \mid w \in R(x) \rbrace, \partial, \llbracket \mathscr{X} \rrbracket)$ with 
\begin{align}
\label{minimaltransition}
	\partial(L)(\alpha) := \begin{cases}
		(p, (\alpha p)^{-1}L) & \textnormal{if } (\alpha p)^{-1}L \not= \emptyset \\
		1 & \textnormal{if } \alpha \in L \\
		0 & \textnormal{otherwise}
	\end{cases},
\end{align}
for $L \in \lbrace w^{-1}\llbracket \mathscr{X} \rrbracket \mid w \in R(x) \rbrace$.
\end{definition}

A few remarks on the well-definedness of above definition are in order. The language accepted by a $G$-automaton is deterministic, and taking the derivative of a language preserves its deterministic nature. Thus only one of the three cases in \eqref{minimaltransition} occurs. Since $\varepsilon \in R(x)$ and $\varepsilon^{-1}L = L$, the initial state of the minimization is well-defined. Transitioning to a new state is well-defined since $v^{-1}(w^{-1}L) = (wv)^{-1}L$.

	It is not hard to see that on a high-level the minimization can be recovered as the image of the final automata homomorphism $\llbracket - \rrbracket: R(\mathscr{X}) \rightarrow \mathcal{L}$, which, as one verifies, satisfies $\llbracket w \rrbracket_{R(\mathscr{X})} = w^{-1}\llbracket \mathscr{X} \rrbracket$. In other words, there exists an epi-mono factorization $R(\mathscr{X}) \twoheadrightarrow m(\mathscr{X}) \hookrightarrow \mathcal{L}$.
	
\subsubsection{Properties of $m(\mathscr{X})$}

\label{properitesofminimal}

In this section we prove properties of $m(\mathscr{X})$, which one would expect to hold by a minimization construction. We begin by showing that minimizing a normal automaton results in a reachable acceptor.

\begin{restatable}{lemma}{reachableminimal}
\label{reachableminimal}
Let $\mathscr{X}$ be a normal $G$-automaton with initial state $x \in X$. Then $\llbracket \mathscr{X} \rrbracket \xrightarrow{w} w^{-1}\llbracket \mathscr{X} \rrbracket$ in $m(\mathscr{X})$ for all $w \in R(x)$. In particular, $m(\mathscr{X})$ is reachable. 
	\end{restatable}

The next result proves that minimizing an automaton preserves its language semantics.

\begin{restatable}{lemma}{minimallanguage}
\label{minimallanguage}
	Let $\mathscr{X}$ be a $G$-automaton, then $\llbracket L \rrbracket = L$ for all $L$ in $m(\mathscr{X})$. In particular, $\llbracket m(\mathscr{X}) \rrbracket = \llbracket \mathscr{X} \rrbracket$.
\end{restatable}

An immediate consequence of above statement is that the states of the minimization can be distinguished by their observable behaviour, that is, different states accept different languages. Another implication of \autoref{minimallanguage} is the normality of the minimization: all states are \emph{live}.

\begin{restatable}{corollary}{minimalnormalobservable}
\label{minimalnormalobservable}
		Let $\mathscr{X}$ be a $G$-automaton, then $m(\mathscr{X})$ is normal and observable.
\end{restatable}

Since $m(\mathscr{X})$ is normal, reachable, and observable, if $\mathscr{X}$ is normal, it is, by our definition, \emph{minimal} (cf. \cite[Def. 15]{van2017calf}). Its size-minimality among normal automata language equivalent to $\mathscr{X}$ can be derived from the abstract definition, cf. \autoref{sizeminimal}.

	\subsubsection{Identifying $m(\mathscr{X})$} 
	\label{identifyingsection}
	
	In this section, we identify the minimization of a normal $G$-automaton with an alternative, but equivalent, construction. In consequence, we are able to derive that the minimization of a normal automaton is size-minimal among language equivalent normal automata and preserves the nesting coequation.
	We begin by observing its universality in the following sense. 
\begin{figure*}
\centering
\begin{subfigure}[b]{0.4\textwidth}
\centering
\adjustbox{scale=0.8,center}{
			\begin{tikzcd}[row sep=small]
		R(\mathscr{X}) \arrow[twoheadrightarrow]{r}{} \arrow[twoheadrightarrow]{dd}{} & r(\mathscr{X}) \arrow[hookrightarrow]{d}{} \arrow[dashed, twoheadrightarrow]{ddl}{\pi} \\
		& \mathscr{X} \arrow{d} \\
		m(\mathscr{X}) \arrow[hookrightarrow]{r} & \mathcal{L}
		\end{tikzcd}	
	}
		\caption{The morphism $\pi$ as unique diagonal.}
				\label{uniquediagonal}
\end{subfigure}
\hfill
\begin{subfigure}[b]{0.55\textwidth}
\centering
\adjustbox{scale=0.8,center}{
			\begin{tikzcd}[row sep=small]
			e \arrow{d} && f \arrow{d}\\
				\mathscr{X}_e \arrow{d} & & \mathscr{X}_e \arrow{d} \\
				\widehat{\mathscr{X}_e} \arrow{d} & & \widehat{\mathscr{X}_f} \arrow{d} \\
				m(\widehat{\mathscr{X}_e}) \arrow{rr}{\cong} & & \arrow{ll} m(\widehat{\mathscr{X}_f})
			\end{tikzcd}
		}
			\caption{$\llbracket e \rrbracket = \llbracket f \rrbracket$ if(f)  $m(\widehat{\mathscr{X}_e})$ and $m(\widehat{\mathscr{X}_f})$ are isomorphic.}
			\label{kozencompletenessdiagram}
\end{subfigure}
\caption{A high-level view of the notions introduced in \autoref{identifyingsection}.}
\end{figure*}
	
\begin{restatable}{proposition}{uniquehomminimal}
\label{uniquehomminimal}
Let $\mathscr{X}$ and $\mathscr{Y}$ be normal $G$-automata with $\llbracket \mathscr{X} \rrbracket =  \llbracket \mathscr{Y} \rrbracket$, and $y \in Y$ the initial state of $\mathscr{Y}$. Then $\pi: r(\mathscr{Y}) \rightarrow m(\mathscr{X})$ with $\pi(z) = w_z^{-1}\llbracket \mathscr{X} \rrbracket$, for $y \xrightarrow{w_z} z$ in $\mathscr{Y}$, is a (surjective) $G$-automata homomorphism, uniquely defined.
\end{restatable}	

The next result shows that the minimization of a normal $G$-automaton is isomorphic to the automaton that arises by identifying semantically equivalent pairs among reachable states. 
\begin{restatable}{lemma}{minimalbisim}
\label{minimalbisim}
	Let $\mathscr{X}$ be a normal $G$-automaton with initial state $x \in X$ and $\pi: r(\mathscr{X}) \twoheadrightarrow m(\mathscr{X})$ as in \autoref{uniquehomminimal}, then $y \simeq z$ if(f) $\pi(y) = \pi(z)$ for all $y, z \in r(x)$. Consequently, $m(\mathscr{X})$ is isomorphic to $r(\mathscr{X})/\simeq$.
	\end{restatable}
	
On a high level, the automata homomorphism $\pi$ can be recovered as the unique (surjective) diagonal making the diagram in \autoref{uniquediagonal} commute.

In \autoref{reachablesubcoalgebra} it was noted that the reachable subautomaton $r(\mathscr{X})$ satisfies the nesting coequation, whenever $\mathscr{X}$ does. By  \autoref{uniquehomminimal} there exists an epimorphism $\pi: r(\mathscr{X}) \twoheadrightarrow m(\mathscr{X})$, if $\mathscr{X}$ is normal. Since coalgebras satisfying a coequation form a covariety, which is closed under homomorphic images \cite{dahlqvist2021write,schmid2021guarded}, we thus can deduce the following result. 

\begin{restatable}{corollary}{minimizationcoequation}
\label{minimizationcoequation}
		Let $\mathscr{X}$ be a normal $G$-automaton, then $m(\mathscr{X})$ satisfies the nesting coequation, whenever $\mathscr{X}$ does.
\end{restatable}

We continue with the observation that two normal $G$-automata are language equivalent if and only if their minimizations are isomorphic. As depicted in \autoref{kozencompletenessdiagram}, this implies that two expressions $e$ and $f$ are language equivalent if and only if the minimizations of their normalized Thompson automata are isomorphic. A similar idea occurs in Kozen's completeness proof for Kleene Algebra \cite[Theorem 19]{kozen1994completeness}.   
	
\begin{restatable}{corollary}{minimalunique}
\label{minimalunique}
	Let $\mathscr{X}$ and $\mathscr{Y}$ be normal $G$-automata, then $\llbracket \mathscr{X} \rrbracket = \llbracket \mathscr{Y} \rrbracket$ if(f) $m(\mathscr{X}) \cong m(\mathscr{Y})$. 
\end{restatable}

We conclude with the size-minimality of the minimization of a normal automaton among language equivalent normal automata.

\begin{restatable}{corollary}{sizeminimal}
\label{sizeminimal}
Let $\mathscr{X}$ and $\mathscr{Y}$ be normal $G$-automata with $\llbracket \mathscr{X} \rrbracket = \llbracket \mathscr{Y} \rrbracket$. Then $\vert m(\mathscr{X}) \vert \leq \vert \mathscr{Y} \vert$, where $\vert m(\mathscr{X}) \vert = \vert \mathscr{Y} \vert$ if(f) $m(\mathscr{X}) \cong \mathscr{Y}$.
\end{restatable}

\section{Learning $m(\mathscr{X})$}	 

\label{learningpart}

In this section we formally investigate the correctness of $\GLStar$ (\autoref{GlStaralgorithm}). Our main result is \autoref{correctnesstheorem}, which shows that if the oracle is instantiated with a deterministic language accepted by a finite normal $G$-automaton $\mathscr{X}$, then  $\GLStar$ terminates with a hypothesis isomorphic to $m(\mathscr{X})$.

For calculations, it will be convenient to use the following definition of an observation table. 
One can show that if the oracle is instantiated with a finite normal $G$-automaton, then one has a well-defined observation table at every step.

\begin{definition}
\label{observationtable}
	An \emph{observation table} $T = (S, E, \row)$ consists of subsets $S \subseteq \GSM, E \subseteq \GS$ and a function $\row: S \cup S \cdot (\At \cdot \Sigma) \rightarrow 2^E$, such that:
	\begin{itemize}
		\item $\varepsilon \in S$ and $\At \subseteq E$ 		
		\item $\alpha p e \in E$ implies $e \in E$ (suffix-closed)	
		\item $s \alpha p \in S$ implies $s \in S$ (prefix-closed)
		\item $s \not = t$ implies $\row(s) \not = \row(t)$ for $s,t \in S$
		\item $\varepsilon \not = s \in S$ implies $\row(s)(e) = 1$ for some $e \in E$
		\item $\row(s\alpha p)(e) = \row(s)(\alpha p e)$, if $\alpha p e \in E$
	\end{itemize}
\end{definition}

Not every table induces a well-defined $G$-automaton. To ensure correctness, we have to restrict ourselves to a  subclass of tables that satisfies two important properties.
We call an observation table \emph{deterministic} if the guarded string language $\row(s) \subseteq \GS$ is deterministic for all $s \in S$. 
An observation table is \emph{closed}, if for all $t \in  S \cdot (\At \cdot \Sigma)$ with $\row(t)(e) = 1$ for some $e \in E$, there exists an $s \in S$ such that $\row(s) = \row(t)$.

\begin{definition}
\label{hypothesis}
Given a closed deterministic observation table $T = (S, E, \row)$, let $m(T) := (\lbrace \row(s) \mid s \in S \rbrace, \delta, \row(\varepsilon))$ be the $G$-automaton with
	\begin{align}
	\label{hypothesisdef}
	\delta(L)(\alpha) = \begin{cases}
		(p, (\alpha p)^{-1}L) & \textnormal{if } (\alpha p)^{-1}L \not = \emptyset \\
		1 & \textnormal{if } \alpha \in L \\
		0 & \textnormal{otherwise}
	\end{cases},
\end{align}
where $L \in \lbrace \row(s) \mid s \in S \rbrace$ and $(\alpha p)^{-1}\row(s) = \row(s \alpha p)$.
\end{definition}

A few remarks on the well-definedness of above definition are in order.
	By \autoref{observationtable} the upper-rows of an observation table are disjoint. Since $T$ is deterministic, precisely one of the three cases in \eqref{hypothesisdef} occurs.
	  If $(\alpha p)^{-1}\row(s)$ is non-empty, there exists, because $T$ is closed, some $t \in S$ with $(\alpha p)^{-1}\row(s) = \row(t)$. This shows that $m(T)$ is closed under transitions.

\subsection{Properties of $m(T)$}

In what follows, let $T$ be a closed deterministic observation table, unless states otherwise. We will establish a few basic properties of $m(T)$. 
First, we observe its reachability, which is implied by a stronger statement. 

\begin{restatable}{lemma}{reachability} 
\label{reachability}
It holds $\row(s) \xrightarrow{t} \row(st)$ in $m(T)$ for all $s \in S$ and $t \in \GSM$, such that $st \in S$. In particular, $m(T)$ is reachable.
\end{restatable}

We call a $G$-automaton $(\mathscr{Y}, y)$ \emph{consistent with $T$}, if $S \subseteq R(y)$ and  $\llbracket y_s \rrbracket(e) = \row(s)(e)$ for all $s \in S$, $e \in E$, and $y_s \in Y$ with $y \xrightarrow{s} y_s$.
By \autoref{reachability}, the automaton $m(T)$ is consistent with $T$ if and only if $\llbracket \row(s) \rrbracket(e) = \row(s)(e)$ for all $s \in S$ and $e \in E$. The consistency of $m(T)$ with $T$ should not be confused with the consistency of $T$ itself. Both terminologies appear frequently in the literature \cite{angluin1987learning}. We show that $m(T)$ is not only consistent with $T$, but has in fact the fewest number of states among all automata consistent with $T$.

\begin{restatable}{lemma}{consistent}
\label{consistent}
	$m(T)$ is size-minimal among automata consistent with $T$.
\end{restatable}

From the consistency of $m(T)$ with $T$ it is straightforward to derive its normality and observability.

\begin{restatable}{lemma}{normalityandobservable}
\label{normalityandobservable}
	$m(T)$ is normal and observable.
\end{restatable}

\subsection{Relationship between $m(T)$ and $m(\mathscr{X})$}

We will next deduce the correctness of $\GLStar$, that is, its termination with an automaton isomorphic to $m(\mathscr{X})$, if the teacher is instantiated with the language accepted by a finite normal automaton $\mathscr{X}$. 

In a first step we establish that any hypothesis admits an injective function from its state-space into the state-space of $m(\mathscr{X})$. The result below does not necessarily require the observation table to be deterministic or closed.

\begin{restatable}{lemma}{injectivefunction}
\label{injectivefunction}
	Let $T = (S, E, \row)$ be an observation table with $\row(t)(e) = \llbracket \mathscr{X} \rrbracket (te)$ for all $t \in S \cup S \cdot (\At \cdot \Sigma),\ e \in E$, and let $x \in X$ be the initial state of $\mathscr{X}$. Then $\pi: \lbrace \row(s) \mid s \in S \rbrace \rightarrow \lbrace w^{-1}\llbracket \mathscr{X} \rrbracket \mid w \in R(x) \rbrace,\ \row(s) \mapsto s^{-1}\llbracket \mathscr{X} \rrbracket$ is a well-defined injective function.
\end{restatable}

If the algorithm terminates with a hypothesis $m(T)$, the latter is, by definition, language equivalent to $\mathscr{X}$, and thus to the minimization $m(\mathscr{X})$, by \autoref{minimallanguage}. The next result implies a stronger statement: in case of termination, the hypothesis $m(T)$ is \emph{isomorphic} to $m(\mathscr{X})$, via the function $\pi$ of \autoref{injectivefunction}.

\begin{restatable}{proposition}{isolanguagequiv}
\label{isolanguagequiv}
	Let $T = (S, E, \row)$ be a closed deterministic observation table with $\row(t)(e) = \llbracket \mathscr{X} \rrbracket(te)$ for all $t \in S \cup S \cdot (\At \cdot \Sigma),\ e \in E$. Let $\pi$ be the injection of \autoref{injectivefunction}, and $\mathscr{X}$ normal. The following are equivalent:
\begin{enumerate}
	\item $\pi: m(T) \simeq m(\mathscr{X})$ is a $G$-automata isomorphism;
	\item $\llbracket m(T) \rrbracket =  \llbracket m(\mathscr{X}) \rrbracket$.
\end{enumerate}
\end{restatable}

The main argument in the proof of \autoref{correctnesstheorem} is  the result below. It shows that if the oracle replies \emph{no} to an equivalence query and provides us with a counterexample $z$, then the table extended with the suffixes of $z$ can immediately be closed only if it is the first time such a situation occurs. 

\begin{restatable}{proposition}{ifclosedthen}
\label{ifclosedthen}
 	Let $T = (S, E, \row)$ be a closed deterministic observation table with $\row(t)(e) = \llbracket \mathscr{X} \rrbracket(te)$ for all $t \in S \cup S \cdot (\At \cdot \Sigma),\ e \in E$. Let $\llbracket m(T) \rrbracket(z) \not = \llbracket \mathscr{X} \rrbracket(z)$ for some $z \in \GS$, and $T' = (S, E \cup \suff(z), \row')$ with $\row'(t)(e) = \llbracket \mathscr{X} \rrbracket(te)$. If $T'$ is closed, then $\row'(\varepsilon)(e) = 0$ for all $e \in E$, but $\row'(\varepsilon)(z') = 1$ for some $z' \in \suff(z)$. 
 	 \end{restatable}
 
 In consequence, an infinite chain of negative equivalence queries and immediately closed extended tables is impossible. Since fixing a closedness defect increases the size of $m(T)$, which by \autoref{injectivefunction} is bounded by the finite number of states in $m(\mathscr{X})$, we can deduce the correctness of \autoref{GlStaralgorithm}.

\begin{restatable}{theorem}{correctnesstheorem}
\label{correctnesstheorem}
	If $\autoref{GlStaralgorithm}$ is instantiated with the language accepted by a finite normal automaton $\mathscr{X}$, then it terminates with a hypothesis isomorphic to $m(\mathscr{X})$.
\end{restatable}

\section{Comparison with Moore automata}

\label{comparisonmoore}

How are the minimal GKAT automaton (\autoref{GLStarmT3}) and the minimal Moore automaton (\autoref{LStarmT4}) representing the language \eqref{acceptedlanguageexample} related? Why should we learn the former, and not the latter? Are there optimizations for $\LStar$ that we could adapt for $\GLStar$? Those are the questions this section seeks to answer. 
 	
\subsection{Embedding of GKAT automata}	

Comparing the GKAT automaton in \autoref{GLStarmT3} with the Moore automaton (with input $\At \cdot \Sigma$ and output $2^{\At}$, short $M$\emph{-automaton}) in \autoref{LStarmT4} suggests that the latter can be recovered from the former by adding a sink-state with which halting transitions can be made explicit. The result below formalises this idea. The language semantics of Moore automata is defined as usual.

\begin{restatable}{lemma}{embeddinglanguage}
\label{embeddinglanguage}
	Given a $G$-automaton $\mathscr{X} = (X, \delta, x)$, let $f(\mathscr{X}) := (X + \lbrace \star \rbrace, \langle \partial, \varepsilon \rangle, x)$ be the $M$-automaton with 
	\begin{gather*}
			\partial(x)(\alpha p) := \begin{cases}
		y & \textnormal{if } x \in X,\ \delta(x)(\alpha) = (p, y) \\
		\star & \textnormal{otherwise}
	\end{cases}
	\qquad 
	\varepsilon(x)(\alpha) := \begin{cases}
		1 & \textnormal{if } x \in X,\ \delta(x)(\alpha) = 1 \\
		0 & \textnormal{otherwise}
	\end{cases}.
		\end{gather*}
	Then $\llbracket x \rrbracket_{\mathscr{X}} = \llbracket x \rrbracket_{f(\mathscr{X})}$ for all $x \in X$, and $\llbracket \star \rrbracket_{f(\mathscr{X})} = \emptyset$. In particular, $\llbracket f(\mathscr{X}) \rrbracket_{f(\mathscr{X})} = \llbracket \mathscr{X} \rrbracket_{\mathscr{X}}$.
\end{restatable}

As one would hope for, above construction maps, up to isomorphism, the minimal GKAT automaton $m(\mathscr{X})$ to the minimal Moore automaton accepting the same language as $\mathscr{X}$.
  
\begin{restatable}{corollary}{minimalembeddingiso}
\label{minimalembeddingiso}
 	 	Let $\mathscr{X}$ be a normal $G$-automaton, then $f(m(\mathscr{X})) \cong m(f(\mathscr{X}))$ as $M$-automata.
 \end{restatable}
 
\subsection{Complexity analysis} 

We now compare the worst-case complexities of $\LStar$ (\autoref{LStaralgorithm}) and $\GLStar$ (\autoref{GlStaralgorithm}) for learning automata representations of GKAT programs $e$. We are mainly interested in a bound to the number of membership queries to $\llbracket e \rrbracket$. The example runs in \autoref{lstarexamplerun} and \autoref{glstarexamplerun} seem to indicate that with respect to this aspect, $\GLStar$ performs better than $\LStar$. The result below confirms this intuition.

\begin{restatable}{proposition}{complexity}
\label{complexity}
 	\autoref{LStaralgorithm} requires at most $O(a * (\vert \At \vert * b))$ many membership queries to $\llbracket e \rrbracket$ for learning a $M$-automaton representation of $e$, whereas $\autoref{GlStaralgorithm} $ requires at most $O(a * (\vert \At \vert + b))$ many membership queries to $\llbracket e \rrbracket$ for learning a $G$-automaton representation of $e$, for some\footnote{Let $m$ be the maximum length of a counterexample and $n$ the size of the minimal Moore automaton accepting $\llbracket e \rrbracket$, then  $a = n * \vert \At \vert * \vert \Sigma \vert$ and $b = m * n$. As \autoref{comparisongraph} shows, $\GLStar$ can be more efficient than $\LStar$ even for small $\vert \At \vert$.} integers $a, b \in \mathbb{N}$. 
 \end{restatable}
 
 One can show that for all integers $x, y$ greater than $2$, the product $x * y$ is strictly greater than the sum $x + y$. Moreover, the difference between $x * y$ and $x + y$ increases with the sizes of $x$ and $y$. The advantage of $\GLStar$ over $\LStar$ for learning deterministic guarded string languages in terms of membership queries thus increases with the size of the set $\At$, which is exponential in the number of primitive tests, $\At \cong 2^{T}$. In applications to network verification, the number of tests, thus atoms, is typically quite large \cite{anderson2014netkat}. The difference between $\GLStar$ and $\LStar$ described in \autoref{complexity} is mainly due to a subtle play with the table indices, based on currying. It can be further increased by avoiding querying certain rows all together, taking into account the deterministic nature of the target language, as indicated in \autoref{glstarexamplerunsection}.
 
 \subsection{Optimized counterexamples}
 
 \label{optimizedcounterexamplesec}
 
 In this section we present an optimization of $\GLStar$ that is based on a subtle refinement of \autoref{ifclosedthen}. We show that, while \autoref{GlStaralgorithm} reacts to a negative equivalence query with counterexample $z \in \GS$ by adding columns for \emph{all} suffixes in $\suff(z)$, it is in fact enough to add columns for a smaller subset of suffixes $\suff(z') \subseteq \suff(z)$, for some $z' \in \suff(z)$ of minimal length. Our approach is inspired by the optimized counterexample handling method of Rivest and Schapire for $\LStar$ \cite{rivest1993inference}.

 \begin{restatable}{lemma}{optimizedcounterexample}
 \label{optimizedcounterexample}
 	Let $T = (S, E, \row)$ be a closed deterministic observation table with $\row(t)(e) = \llbracket \mathscr{X} \rrbracket(te)$ for all $t \in S \cup S \cdot (\At \cdot \Sigma),\ e \in E$. Let $\llbracket m(T) \rrbracket(z) \not = \llbracket \mathscr{X} \rrbracket(z)$ for some $z \in \GS$, and 
 	$z' := \textnormal{min}(A_z)$\footnote{
   $
A_z := \lbrace z' \in \suff(z) \mid z = v \alpha p z',\ \row(\varepsilon) \xrightarrow{v} \row(s_v),\ x \xrightarrow{s_v} x_{s_v},\ \llbracket \row(s_v) \rrbracket(\alpha p z') \not = \llbracket x_{s_v} \rrbracket(\alpha p z') \rbrace
$
}. If $T' = (S, E \cup \suff(z'), \row')$ with $\row'(t)(e) = \llbracket \mathscr{X} \rrbracket(te)$ is closed, then $\row'(\varepsilon)(e) = 0$ for all $e \in E$, but $\row'(\varepsilon)(z') = 1$.
 \end{restatable}
  Let $z_0$ be the shortest suffix of $z$ and $z_i$ the suffix of $z$ of length $\vert z_{i-1}\vert + 1$. The suffix $\textnormal{min}(A_z)$ can easily be computed in at most $\vert \suff(z) \vert - 1$ steps: verify whether $z_i \in A_z$, beginning with $z_0$; if positive, break and set $\textnormal{min}(A_z) := z_i$, otherwise loop with $z_{i+1}$.
  
 For example, if $T$ is the closed table in \autoref{GLStarT2} with the corresponding hypothesis $m(T)$ in \autoref{GLStarmT2} and counterexample $z = b p \overline{b} q b$, then $z' = \textnormal{min}(A_z) = \overline{b}qb$, since $b \not \in A_z$. \autoref{optimizedcounterexample} shows that, instead of adding columns for the two non-present suffixes $b p \overline{b} q b$ and  $\overline{b}qb$ of $z$, it is sufficient to add only one column for the single non-present suffix  $\overline{b}qb$ of $z'$. In this case, the counterexample $z$ is relatively short, thus the number of avoided columns small; in general, however, the advantage can be more significant.

 \section{Implementation}

 \begin{figure}
 \centering	
\begin{subfigure}[c]{.48\textwidth}
\centering
\begin{subfigure}[c]{.48\textwidth}
 \adjustbox{scale=0.55,center}{
	 	\begin{tikzpicture}
\begin{axis}[
	width = 22em,
    xlabel={$\vert T \vert = \vert \lbrace t_1,...,t_n \rbrace \vert$},
    ylabel={Membership queries to $\llbracket e \rrbracket$},
    xmin=1, xmax=9,
    ymin=1, ymax=8000000,
    xtick={1,2,3,4,5,6,7,8,9},
    ytick={500000,1000000,2500000,5000000,8000000},
    legend pos=north west,
]
 \addplot[
    color=red,
    mark=x,
    ]
    coordinates {
    (1,114)
    (2,444)
    (3,1752)
    (4,6960)
    (5,27744)
    (6,110784)
    (7,442752)
    (8,1770240)
    (9,7079424)
    };
\addplot[
    color=blue,
    mark=square,
    ]
    coordinates {
    (1,26)
    (2,100)
    (3,392)
    (4,1552)
    (5,6176)
    (6,24640)
    (7,98432)
    (8,393472)
    (9,1573376)
    };
    \legend{$\LStar$,$\GLStar$}   
\end{axis}
\end{tikzpicture}
}
\end{subfigure}
\begin{subfigure}[c]{.48\textwidth}
 \adjustbox{scale=0.55,center}{
\begin{tabular}{ c | c | c }
\specialrule{.1em}{.15em}{.15em} 
 $\vert T \vert$ & $\GLStar$ & $\LStar$ \\
 \specialrule{.05em}{.15em}{.15em}
1 & 26 & 114 \\
2 & 100 & 444 \\ 
3 & 392 & 1.752 \\
4 & 1.552 & 6.960 \\
5 & 6.176 & 27.744 \\
6 & 24.640 & 110.784 \\
7 & 98.432 & 442.752 \\
8 & 393.472 & 1.770.240 \\
9 & 1.573.376 & 7.079.424	\\
 \specialrule{.1em}{.15em}{.15em} 
\end{tabular}
}
\end{subfigure}
\caption{$e = \iif\ t_1\ \tthen\ \ddo\ p_1\ \eelse\ \ddo\ p_2$}
\label{performance_ife}
\end{subfigure}
\begin{subfigure}[c]{.48\textwidth}
\centering
\begin{subfigure}[c]{.48\textwidth}
 \adjustbox{scale=0.55,center}{
	 	\begin{tikzpicture}
\begin{axis}[
	width = 22em,
    xlabel={$\vert T \vert = \vert \lbrace t_1,...,t_n \rbrace \vert$},
    ylabel={Membership queries to $\llbracket e \rrbracket$},
    xmin=1, xmax=9,
    ymin=1, ymax=8000000,
    xtick={1,2,3,4,5,6,7,8,9},
    ytick={500000,1000000,2500000,5000000, 8000000},
    legend pos=north west,
]
 \addplot[
    color=red,
    mark=x,
    ]
    coordinates {
    (1,78)
    (2,300)
    (3,1176)
    (4,4659)
    (5,18528)
    (6,73920)
    (7,295296)
    (8,1180416)
    (9,4720128)
    };
    
\addplot[
    color=blue,
    mark=square,
    ]
    coordinates {
    (1,36)
    (2,102)
    (3,330)
    (4,1170)
    (5,4386)
    (6,16962)
    (7,66690)
    (8,264450)
    (9,1053186)
    };
    \legend{$\LStar$,$\GLStar$}   
\end{axis}
\end{tikzpicture}
}
\end{subfigure}
\begin{subfigure}[c]{.48\textwidth}
 \adjustbox{scale=0.55,center}{
\begin{tabular}{ c | c | c }
\specialrule{.1em}{.15em}{.15em} 
 $\vert T \vert$ & $\GLStar$ & $\LStar$ \\
 \specialrule{.05em}{.15em}{.15em}
1 & 36 & 78 \\
2 & 102 & 300 \\ 
3 & 330 & 1.176 \\
4 & 1.170 & 4.656 \\
5 & 4.386 & 18.528 \\
6 & 16.962 & 73.920 \\
7 & 66.690 & 295.296 \\
8 & 264.450 & 1.180.416 \\
9 & 1.053.186 & 4.720.128	\\
 \specialrule{.1em}{.15em}{.15em} 
\end{tabular}
}
\end{subfigure}
\caption{$e =  (\wwhile\ t_1\ \ddo\ p_1); \ddo\ p_2$}
\label{performance_while}
\end{subfigure}
\caption{A comparison between $\GLStar$ and $\LStar$ with respect to membership queries.}
\label{comparisongraph}
 \end{figure}

We have implemented both $\GLStar$ and $\LStar$ in OCaml; the code is available on GitHub\footnote{\url{https://github.com/zetzschest/gkat-automata-learning}}. The implementation allows one to compare, for any GKAT expression $e \in \Exp_{\Sigma, T}$, the number of membership queries to $\llbracket e \rrbracket$ required by $\GLStar$ for learning a $G$-automaton representation of $e$, with the number of membership queries to $\llbracket e \rrbracket$ required by $\LStar$ for learning a $M$-automaton representation of $e$. For each run, we output, for both algorithms, a trace of the involved hypotheses as tables in the \texttt{.csv} format and graphs in the \texttt{.dot} format, as well as an overview of the numbers of involved queries in the \texttt{.csv} format. 

In \autoref{performance_ife} we present the results for the expression $e =  \iif\ t_1\ \tthen\ \ddo\ p_1\ \eelse\ \ddo\ p_2$, the primitive actions $\Sigma = \lbrace p_1, p_2, p_3 \rbrace$, and primitive tests $T = \lbrace t_1,...,t_n \rbrace$ parametric in $n = 1,...,9$. We find that $\GLStar$ outperforms $\LStar$ for all choices of $n$. The difference in the number of membership queries increases with the size of $n$, as suggested by \autoref{complexity}. For $n = 9$ the number of atoms is $2^9$, resulting in an already relatively large number of queries for both algorithms.
The picture is similar in \autoref{performance_while}, where we choose the expression $e =  (\wwhile\ t_1\ \ddo\ p_1); \ddo\ p_2$, the primitive actions $\Sigma = \lbrace p_1, p_2 \rbrace$, and primitive tests $T = \lbrace t_1,...,t_n \rbrace$ parametric in $n = 1,...,9$. Again, $\GLStar$ requires significantly less queries in all cases of $n$, and the difference increases with the size of $n$.

Our implementation generates an oracle for $\LStar$ from a GKAT expression $e$ in the following way. First, we interpret $e$ as a KAT expression $\iota(e)$ via the standard embedding of GKAT into KAT. Next, we generate from the latter a Moore automaton $\mathscr{X}_{\iota(e)}$ accepting $\llbracket e \rrbracket$, by using Kozen's syntactic Brzozowksi derivatives for KAT \cite{kozen2017coalgebraic}. Finally, we answer an equivalence query from a Moore automaton $\mathscr{Y}$ by running a bisimulation between $\mathscr{X}_{\iota(e)}$ and $\mathscr{Y}$, similarly to \cite[Fig. 1]{pous2015symbolic}, and a membership query from $w \alpha \in \GS$ by returning the value of $\alpha$ at the output of the state in $\mathscr{X}_{\iota(e)}$ reached by $w$, that is, $\llbracket e \rrbracket(w \alpha)$. A membership query from $w \in \GSM$ is answered by querying $w \alpha \in \GS$ for all $\alpha \in \At$.

With the oracle for $\LStar$, we can derive an oracle for $\GLStar$ as follows. Membership queries $w \alpha \in \GS$ are delegated and answered by the oracle for $\LStar$ as explained above. An equivalence query from a GKAT automaton $\mathscr{Y}$ is answered by posing an equivalence query to the oracle for $\LStar$ with the Moore automaton $f(\mathscr{Y})$ obtained via the embedding defined in \autoref{embeddinglanguage}. If the oracle for $\LStar$ replies with a counterexample $z \in \GSM$, we extend $z$ with an $\alpha \in \At$ such that $\llbracket  \mathscr{Y} \rrbracket(z\alpha) \not = \llbracket e \rrbracket(z\alpha)$.

\section{Related work}
\label{relatedwork}

GKAT is a variation on KAT \cite{kozen1996kleene} that one obtains by restricting the union and iteration operations from KAT to guarded versions. While GKAT is less expressive than KAT, term equivalence is notably more efficiently decidable \cite{smolka2019guarded,kozen1996kleene}, making it a candidate for the foundations of network-programming \cite{smolka2019scalable,anderson2014netkat,foster2015coalgebraic}

GKAT automata appear in the literature already prior to \cite{smolka2019guarded}, e.g. in the work of Kozen \cite{kozen2008bohm} under the name \emph{strictly deterministic automata}. In the latter, Kozen states that GKAT automata correspond to a limited class of \emph{automata with guarded strings (AGS)} \cite{kozen2001automata}, for which he gives determinization and minimization constructions. In a different paper \cite{kozen2017coalgebraic} Kozen introduces a second definition of (deterministic) AGS as Moore automata, and states the difference to the definition of AGS in \cite{kozen2001automata} is inessential. 

Recently, a new perspective on the semantics and coalgebraic theory of GKAT has been given in terms of coequations \cite{schmid2021guarded,dahlqvist2021write}. Using the Thompson construction, it is possible to construct for every expression $e$ a language equivalent automaton $\mathscr{X}_e$. In \cite{kozen2008bohm} it was shown that the inverse does generally not hold: there exists a GKAT automaton that is inequivalent to $\mathscr{X}_e$ for all expressions $e$. In consequence, \cite{smolka2019guarded} proposed a subclass of \emph{well-nested} automata and showed that every finite well-nested automaton is bisimilar to $\mathscr{X}_e$ for some $e$. In \cite{schmid2021guarded} it was shown that well-nestedness is in fact too restrictive: there exists an automaton that is bisimilar to $\mathscr{X}_e$ for some $e$, but not well-nested. 
To capture the \emph{full} class of automata exhibiting the behaviour of expressions, one has to extend the class of well-nested automata to the class of automata satisfying the \emph{nesting coequation}, which forms a \emph{covariety} \cite{dahlqvist2021write}. 

Active automata learning is a technique used for deriving a model from a black-box by interacting with it via observations. The seminal algorithm  $\LStar$\cite{angluin1987learning} learns deterministic finite automata, but since then has been extended to other classes of automata \cite{angluin1997learning,aarts2010learning,moerman2017learning}, including Moore automata. Typically, algorithms such as $\LStar$ are designed to output for a given language a unique minimal acceptor. Not all classes admit a canonical minimal acceptor, for instance, learning non-deterministic models is a challenge \cite{denis2001residual,bollig2009angluin,zetzsche2021,van2020learning}. 

\section{Discussion and future work}

We have presented $\GLStar$, an algorithm for learning the GKAT automaton representation of a black-box, by observing its behaviour via queries to an oracle. We have shown that for every normal GKAT automaton there exists a unique size-minimal normal automaton, accepting the same language: its minimization. We have identified the minimization with an alternative but equivalent construction, and derived its preservation of the nesting coequation. A central result showed that if the oracle in $\GLStar$ is instantiated with the language accepted by a finite normal automaton, then $\GLStar$ terminates with its minimization. A complexity analysis showed the advantage of $\GLStar$ over $\LStar$ for learning automata representations of GKAT programs in terms of membership queries. We discussed additional optimizations, and implemented $\GLStar$ and $\LStar$ in OCaml to compare their performances on example programs. 

There are numerous directions in which the present work could be further explored.
In \autoref{optimizedcounterexamplesec} we introduced an optimization for $\GLStar$ which is inspired by Rivest and Schapire's counterexample handling method for $\LStar$ \cite{rivest1993inference}. The \textit{oberservation pack} algorithm for $\LStar$ \cite{howar2012active} has successfully combined Rivest and Schapire's method with an efficient \textit{discrimination tree} data structure \cite{kearns1994introduction}. The state-of-the-art \textit{TTT}-algorithm \cite{isberner2014ttt} for $\LStar$ extends the former with discriminator finalization techniques. It thus is natural to ask whether for $\GLStar$ there exist similarly efficient data structures, potentially exploiting the deterministic nature of the languages accepted by GKAT automata.

While $\LStar$ has seen major improvements over the years and has inspired numerous variations for different types of transition systems, all approaches remain in common their focus on the \emph{equivalence} of observations. The recently presented $\Lsharp$ algorithm \cite{vaandrager2021new} takes a different perspective: it instead focuses on \emph{apartness}, a constructive form of inequality. $\Lsharp$ does not require data-structures such as observation tables or discrimination trees, instead operating directly on tree-shaped automata. It remains open whether a similar shift in perspective is feasible for $\GLStar$.

There exist various domain-specific extensions of KAT (e.g. KAT+B! \cite{grathwohl2014kat+}, NetKAT \cite{anderson2014netkat}, ProbNetKAT \cite{foster2016probabilistic}), and similar directions have been proposed for GKAT. In particular, it has been noted that GKAT is better fit for probabilistic domains than KAT, as it avoids mixing non-determinism with probabilities \cite{smolka2019scalable}. We expect that in the future, for such extensions of GKAT, there will be interest in developing the corresponding automata (learning) theories.

\bibliography{example.bib}

\begin{thebibliography}{10}
\expandafter\ifx\csname url\endcsname\relax
  \def\url#1{\texttt{#1}}\fi
\expandafter\ifx\csname urlprefix\endcsname\relax\def\urlprefix{ }\fi
\newcommand{\enquote}[1]{``#1''}

\bibitem{aarts2010learning}
Aarts, F. and F.~Vaandrager, \emph{Learning i/o automata}, in: \emph{International Conference on Concurrency Theory}, Springer, 2010, pp. 71--85.
\newline\url{https://doi.org/10.1007/978-3-642-15375-4_6}

\bibitem{anderson2014netkat}
Anderson, C.~J., N.~Foster, A.~Guha, J.-B. Jeannin, D.~Kozen, C.~Schlesinger and D.~Walker, \emph{Netkat: Semantic foundations for networks}, ACM Sigplan Notices \textbf{49} (2014), pp.~113--126.
\newline\url{https://doi.org/10.1145/2578855.2535862}

\bibitem{angluin1987learning}
Angluin, D., \emph{Learning regular sets from queries and counterexamples}, Information and computation \textbf{75} (1987), pp.~87--106.
\newline\url{https://doi.org/10.1016/0890-5401(87)90052-6}

\bibitem{angluin1997learning}
Angluin, D. and M.~Cs{\H{u}}r{\"o}s, \emph{Learning markov chains with variable memory length from noisy output}, in: \emph{Proceedings of the tenth annual conference on Computational learning theory}, 1997, pp. 298--308.
\newline\url{https://doi.org/10.1145/267460.267517}

\bibitem{bollig2009angluin}
Bollig, B., P.~Habermehl, C.~Kern and M.~Leucker, \emph{Angluin-style learning of nfa}, in: \emph{Twenty-First International Joint Conference on Artificial Intelligence}, 2009.
\newline\url{https://dl.acm.org/doi/10.5555/1661445.1661605}

\bibitem{chalupar2014automated}
Chalupar, G., S.~Peherstorfer, E.~Poll and J.~De~Ruiter, \emph{Automated reverse engineering using lego{\textregistered}}, in: \emph{8th $\{$USENIX$\}$ Workshop on Offensive Technologies ($\{$WOOT$\}$ 14)}, 2014.\newline\url{https://www.usenix.org/conference/woot14/workshop-program/presentation/chalupar}

\bibitem{dahlqvist2021write}
Dahlqvist, F. and T.~Schmid, \emph{How to write a coequation ((co) algebraic pearls)}, in: \emph{9th Conference on Algebra and Coalgebra in Computer Science (CALCO 2021)}, Schloss Dagstuhl-Leibniz-Zentrum f{\"u}r Informatik, 2021.\newline\url{https://doi.org/10.4230/LIPIcs.CALCO.2021.13}

\bibitem{de2015protocol}
De~Ruiter, J. and E.~Poll, \emph{Protocol state fuzzing of $\{$TLS$\}$ implementations}, in: \emph{24th USENIX Security Symposium (USENIX Security 15)}, 2015, pp. 193--206.
\newline\url{https://dl.acm.org/doi/10.5555/2831143.2831156}

\bibitem{denis2001residual}
Denis, F., A.~Lemay and A.~Terlutte, \emph{Residual finite state automata}, in: \emph{Annual Symposium on Theoretical Aspects of Computer Science}, Springer, 2001, pp. 144--157.
\newline\url{https://link.springer.com/chapter/10.1007/3-540-44693-1_13}

\bibitem{feamster2014road}
Feamster, N., J.~Rexford and E.~Zegura, \emph{The road to sdn: an intellectual history of programmable networks}, ACM SIGCOMM Computer Communication Review \textbf{44} (2014), pp.~87--98.\newline\url{https://doi.org/10.1145/2602204.2602219}

\bibitem{foster2016probabilistic}
Foster, N., D.~Kozen, K.~Mamouras, M.~Reitblatt and A.~Silva, \emph{Probabilistic netkat}, in: \emph{European Symposium on Programming}, Springer, 2016, pp. 282--309.
\newline\url{https://link.springer.com/chapter/10.1007/978-3-662-49498-1_12}

\bibitem{foster2015coalgebraic}
Foster, N., D.~Kozen, M.~Milano, A.~Silva and L.~Thompson, \emph{A coalgebraic decision procedure for netkat}, in: \emph{Proceedings of the 42nd Annual ACM SIGPLAN-SIGACT Symposium on Principles of Programming Languages}, 2015, pp. 343--355.
\newline{https://doi.org/10.1145/2676726.2677011}

\bibitem{grathwohl2014kat+}
Grathwohl, N. B.~B., D.~Kozen and K.~Mamouras, \emph{Kat+ b!}, in: \emph{Proceedings of the joint meeting of the twenty-third {EASCL} annual conference on {C}omputer {S}cience {L}ogic ({CSL}) and the twenty-ninth annual {ACM/IEEE} {S}ymposium on {L}ogic in {C}omputer {S}cience ({LiCS})}, 2014, pp. 1--10.
\newline\url{https://doi.org/10.1145/2603088.2603095}

\bibitem{hagerer2002model}
Hagerer, A., H.~Hungar, O.~Niese and B.~Steffen, \emph{Model generation by moderated regular extrapolation}, in: \emph{International Conference on Fundamental Approaches to Software Engineering}, Springer, 2002, pp. 80--95.\newline\url{https://link.springer.com/chapter/10.1007/3-540-45923-5_6}

\bibitem{howar2012active}
Howar, F., \enquote{Active learning of interface programs.} Ph.D. thesis, Dortmund University of Technology (2012).\newline\url{http://doi.org/10.17877/DE290R-4817}

\bibitem{isberner2014ttt}
Isberner, M., F.~Howar and B.~Steffen, \emph{The {TTT} algorithm: a redundancy-free approach to active automata learning}, in: \emph{International Conference on Runtime Verification}, Springer, 2014, pp. 307--322.\newline\url{https://link.springer.com/chapter/10.1007/978-3-319-11164-3_26}

\bibitem{kearns1994introduction}
Kearns, M.~J., U.~V. Vazirani and U.~Vazirani, \enquote{An introduction to computational learning theory,} MIT press, 1994. ISBN 9780262111935.

\bibitem{kleene1951representation}
Kleene, S., \emph{Representation of events in nerve nets and finite automata}, Automata studies \textbf{3} (1951), p.~41.\newline\url{https://doi.org/10.1515/9781400882618-002}

\bibitem{kozen1994completeness}
Kozen, D., \emph{A completeness theorem for {Kleene} algebras and the algebra of regular events}, Information and computation \textbf{110} (1994), pp.~366--390.\newline\url{https://doi.org/10.1006/inco.1994.1037}

\bibitem{kozen1997kleene}
Kozen, D., \emph{{Kleene} algebra with tests}, ACM Transactions on Programming Languages and Systems (TOPLAS) \textbf{19} (1997), pp.~427--443.\newline
\url{https://doi.org/10.1145/256167.256195}

\bibitem{kozen2001automata}
Kozen, D., \emph{Automata on guarded strings and applications}, Technical report, Cornell University (2001).\newline\url{https://www.cs.cornell.edu/~kozen/Papers/ags.pdf}

\bibitem{kozen2017coalgebraic}
Kozen, D., \emph{On the coalgebraic theory of kleene algebra with tests}, in: \emph{Rohit Parikh on Logic, Language and Society}, Springer, 2017 pp. 279--298.
\newline\url{https://doi.org/10.1007/978-3-319-47843-2_15}

\bibitem{kozen1996kleene}
Kozen, D. and F.~Smith, \emph{Kleene algebra with tests: Completeness and decidability}, in: \emph{International Workshop on Computer Science Logic}, Springer, 1996, pp. 244--259.
\newline\url{https://link.springer.com/chapter/10.1007/3-540-63172-0_43}

\bibitem{kozen2008bohm}
Kozen, D. and W.-L.~D. Tseng, \emph{The {B}{\"o}hm--{J}acopini theorem is false, propositionally}, in: \emph{International Conference on Mathematics of Program Construction}, Springer, 2008, pp. 177--192.\newline
\url{https://doi.org/10.1007/978-3-540-70594-9_11}

\bibitem{maler1995learnability}
Maler, O. and A.~Pnueli, \emph{On the learnability of infinitary regular sets}, Information and Computation \textbf{118} (1995), pp.~316--326.
\newline\url{https://doi.org/10.1006/inco.1995.1070}

\bibitem{moerman2017learning}
Moerman, J., M.~Sammartino, A.~Silva, B.~Klin and M.~Szynwelski, \emph{Learning nominal automata}, in: \emph{Proceedings of the 44th ACM SIGPLAN Symposium on Principles of Programming Languages}, 2017, pp. 613--625.
\newline\url{https://doi.org/10.1006/inco.1995.1070}

\bibitem{moore1956gedanken}
Moore, T., \emph{Gedanken--experiments on Sequential Machines}, in: \emph{Sequential Machines, Automata Studies, Annals of Mathematical Studies, no. 34}, Citeseer, 1956.
\newline\url{https://doi.org/10.1515/9781400882618-006}

\bibitem{nerode1958linear}
Nerode, A., \emph{Linear automaton transformations}, Proceedings of the American Mathematical Society \textbf{9} (1958), pp.~541--544.
\newline\url{https://doi.org/10.2307/2033204}

\bibitem{pous2015symbolic}
Pous, D., \emph{Symbolic algorithms for language equivalence and kleene algebra with tests}, in: \emph{Proceedings of the 42nd Annual ACM SIGPLAN-SIGACT Symposium on Principles of Programming Languages}, 2015, pp. 357--368.
\newline\url{https://doi.org/10.1145/2775051.2677007}

\bibitem{rivest1993inference}
Rivest, R.~L. and R.~E. Schapire, \emph{Inference of finite automata using homing sequences}, Information and Computation \textbf{103} (1993), pp.~299--347.
\newline\url{https://doi.org/10.1006/inco.1993.1021}

\bibitem{schmid2021guarded}
Schmid, T., T.~Kapp{\'e}, D.~Kozen and A.~Silva, \emph{Guarded kleene algebra with tests: Coequations, coinduction, and completeness}, arXiv preprint  (2021).
\newline\url{https://doi.org/10.48550/arXiv.2102.08286}

\bibitem{smolka2019guarded}
Smolka, S., N.~Foster, J.~Hsu, T.~Kapp{\'e}, D.~Kozen and A.~Silva, \emph{Guarded kleene algebra with tests: verification of uninterpreted programs in nearly linear time}, Proceedings of the ACM on Programming Languages \textbf{4} (2019), pp.~1--28.
\newline\url{https://doi.org/10.1145/3371129}

\bibitem{smolka2019scalable}
Smolka, S., P.~Kumar, D.~M. Kahn, N.~Foster, J.~Hsu, D.~Kozen and A.~Silva, \emph{Scalable verification of probabilistic networks}, in: \emph{Proceedings of the 40th ACM SIGPLAN Conference on Programming Language Design and Implementation}, 2019, pp. 190--203.
\newline\url{https://doi.org/10.1145/3314221.3314639}

\bibitem{thompson1968programming}
Thompson, K., \emph{Programming techniques: Regular expression search algorithm}, Communications of the ACM \textbf{11} (1968), pp.~419--422.
\newline\url{https://doi.org/10.1145/363347.363387}

\bibitem{vaandrager2017model}
Vaandrager, F., \emph{Model learning}, Communications of the ACM \textbf{60} (2017), pp.~86--95.
\newline\url{https://doi.org/10.1145/2967606}

\bibitem{vaandrager2021new}
Vaandrager, F., B.~Garhewal, J.~Rot and T.~Wi{\ss}mann, \emph{A new approach for active automata learning based on apartness}, arXiv preprint arXiv:2107.05419  (2021).
\newline\url{https://doi.org/10.48550/arXiv.2107.05419}

\bibitem{van2017calf}
van Heerdt, G., M.~Sammartino and A.~Silva, \emph{Calf: Categorical automata learning framework}, Computer Science Logic 2017  (2017).
\newline
\url{https://drops.dagstuhl.de/opus/volltexte/2017/7695/pdf/LIPIcs-CSL-2017-29.pdf}

\bibitem{van2020learning}
van Heerdt, G., M.~Sammartino and A.~Silva, \emph{Learning automata with side-effects}, in: \emph{International Workshop on Coalgebraic Methods in Computer Science}, Springer, 2020, pp. 68--89.\newline
\url{https://link.springer.com/chapter/10.1007/978-3-030-57201-3_5}

\bibitem{zetzsche2021}
Zetzsche, S., G.~van Heerdt, M.~Sammartino and A.~Silva, \emph{Canonical automata via distributive law homomorphisms}, Electronic Proceedings in Theoretical Computer Science \textbf{351} (2021), p.~296–313.
\newline\urlprefix\url{http://doi.org/10.4204/EPTCS.351.18}

\end{thebibliography}
\bibliographystyle{entics}

\end{document}